\documentclass[aps,pra,reprint,superscriptaddress,nofootinbib,longbibliography,floatfix,showkeys]{revtex4-2}

\usepackage[T1]{fontenc}
\usepackage[utf8]{inputenc}
\usepackage{amsmath,amsfonts,amssymb,amsthm,amsxtra}
\usepackage{physics}
\usepackage{graphicx}
\usepackage{bm}
\usepackage{xcolor}
\usepackage{algorithm}
\usepackage{algpseudocode}
\usepackage[bookmarks=false,colorlinks=true,urlcolor=blue,citecolor=blue,linkcolor=blue]{hyperref}

\newtheorem{theorem}{Theorem}
\newtheorem{lemma}{Lemma}

\newcommand{\be}{\begin{equation}}
\newcommand{\ee}{\end{equation}}
\newcommand{\bea}{\begin{eqnarray}}
\newcommand{\eea}{\end{eqnarray}}
\newcommand{\tO}{\textbf{O}}
\newcommand{\tx}{\textbf{x}}
\newcommand{\tX}{\textbf{X}}
\newcommand{\tu}{\textbf{u}}
\newcommand{\tK}{\textbf{K}}
\newcommand{\tk}{\textbf{k}}
\newcommand{\tf}{\textbf{f}}
\newcommand{\bPhi}{\bm{\Phi}}
\newcommand{\bLambda}{\bm{\Lambda}}

\graphicspath{{figure/}{../figures/}}

\begin{document}

\title{A hybrid method for quantum dynamics simulation}

\author{Niladri Gomes}
\email{nldrgomeser48@gmail.com}
\affiliation{Applied Mathematics and Computational Research Division, Lawrence Berkeley National Laboratory, Berkeley, California 94720, USA}

\author{Jia Yin}
\email{jiayin@fudan.edu.cn}
\affiliation{School of Mathematical Sciences, Fudan University, Shanghai 200433, China}

\author{Siyuan Niu}
\affiliation{Applied Mathematics and Computational Research Division, Lawrence Berkeley National Laboratory, Berkeley, California 94720, USA}

\author{Chao Yang}
\affiliation{Applied Mathematics and Computational Research Division, Lawrence Berkeley National Laboratory, Berkeley, California 94720, USA}

\author{Wibe Albert de Jong}
\affiliation{Applied Mathematics and Computational Research Division, Lawrence Berkeley National Laboratory, Berkeley, California 94720, USA}

\date{\today}

\begin{abstract}
We propose a hybrid approach to simulate quantum many-body dynamics by combining Trotter-based quantum algorithm with classical dynamic mode decomposition. In quantum simulation, the interest often lies in estimating observables rather than explicitly obtaining the wave function. Instead of characterizing a quantum state, our method predicts observables of the quantum state in the long time by using data from a set of short-time measurements from a quantum computer. The upper bound for the global error of our method scales as $\mathcal{O}(t^{3/2})$ with a fixed set of the measurement. We apply our method to quench dynamics in Hubbard model and nearest neighbor spin systems and show that the observable properties can be predicted up to a reasonable error by controlling the number of data points obtained from the quantum measurements.
\end{abstract}


\maketitle
\section{Introduction}
Simulating quantum dynamics is becoming increasingly viable with the advancements in universal quantum computers \cite{Smith_2019, AVQDS, lee2022variational}. 
Quantum computers offer the opportunity to explore previously inaccessible physics, such as simulating lattice gauge theories with dynamical gauge fields and strongly-correlated quantum systems with multi-body couplings.
Despite the advancements, the current state of the art devices still suffer from significant noise and imperfections~\cite{nisq}. There is a significant ongoing debate about whether the currently available processors can be sufficiently improved in terms of reliability to execute shorter-depth quantum circuits on a scale that would offer practical advantages for problem-solving~\cite{lee2208there}. Presently, the general consensus suggests that the implementation of even basic quantum circuits capable of surpassing classical capabilities will likely require the arrival of more advanced, fault-tolerant processors.

Consequently, a crucial area of research focuses on discovering methodologies to conduct meaningful scientific computations on these systems. As a result, variational algorithms have emerged as a valuable approach to investigate energy eigenstates and dynamics within spin and fermionic systems \cite{variational_mclachlan, vqe, vqe_theory, vqe_pea_h2, hardware_efficient_vqe, VQE_qcc, FengVQE}. 
However, beyond the assessment of energy eigenstates, there is a growing interest within the scientific community to explore the dynamic properties of electronic matter at low temperatures, presenting an immediate area of interest.
With the existing NISQ (Noisy Intermediate-Scale Quantum) resources, it is feasible to leverage a combination of existing classical and numerical tools to achieve valuable outcomes. To that end, predictive techniques for a large number of obeservables to characterize quantum states are becoming popular with the advent of classical shadow methods \cite{huang2020predicting} and its derandomized version \cite{huang2021efficient}. In our current work, we will adopt a different predictive protocol to estimate observables during a long time Hamiltonian simulation.

To retrieve the dynamical properties of a system described by a wavefunction $\ket{\psi}$ at an initial time, and subject to a Hamiltonian $H$, one can calculate $e^{-iHt}\ket{\psi}$ to obtain the state at a time $t$.
Different methods have been proposed in the recent years to do Hamiltonian simulation using a quantum computer. Lloyd's method for simulating  assumes a tensor product structure of smaller subsystems \cite{lloyd1996universal}. Aharonov and Ta-Shma (ATS) \cite{aharonov2003adiabatic} consider the case of sparse Hamiltonian where it can be simulated by efficiently calculating the nonzero entries in a given column of the Hamiltonian. In recent years, variational approaches have shown promising results at generating low depth circuits \cite{AVQDS,lee2022variational,variational_mclachlan,Steckmann2021}. However, the variational methods lack a definitive proof of scaling and error bounds in the assymptotic limit. Considering efficient scaling, Hamiltonian simulation using first-order integrators of Suzuki~\cite{berry2007efficient} are one of the promising choices as scalable quantum algorithms.

Despite their favorable asymptotic scaling, the Trotter-based methods suffer from increasing circuit depth which makes it hard if not impossible to do any meaningful calculation for long time simulations.
Although recent advancements have led to impressive scalability in terms of quantum volume, enabling the consideration of problems with a large number of qubits, meaningful results can only be obtained for a limited number of time steps. 
One possible approach is to utilize established classical extrapolation methods such as dynamic mode decomposition (DMD) to predict long time state properties collected from the hardware data. However, due to the complex nature of quantum wave functions, applying such predictions can be computationally intensive.
In scientific problems, the interest often lies in estimating observables rather than explicitly obtaining the wave function's form.  In this paper, instead of employing DMD for state vector extrapolation, we propose a method that predicts observables of a quantum state by gathering data from a quantum computer. Compared to the neural network-based methods for the simulation of quantum dynamics through observables~\cite{mohseni2022deep,mohseni2024deep}, our method offers significant computational advantages as it does not require generating training trajectories from random processes. Moreover, error analysis for neural network techniques remains impractical, while our approach allows for an estimation of error bounds, demonstrating the reliability of DMD.

\begin{figure*}[t]
    \includegraphics[width=0.7\linewidth]{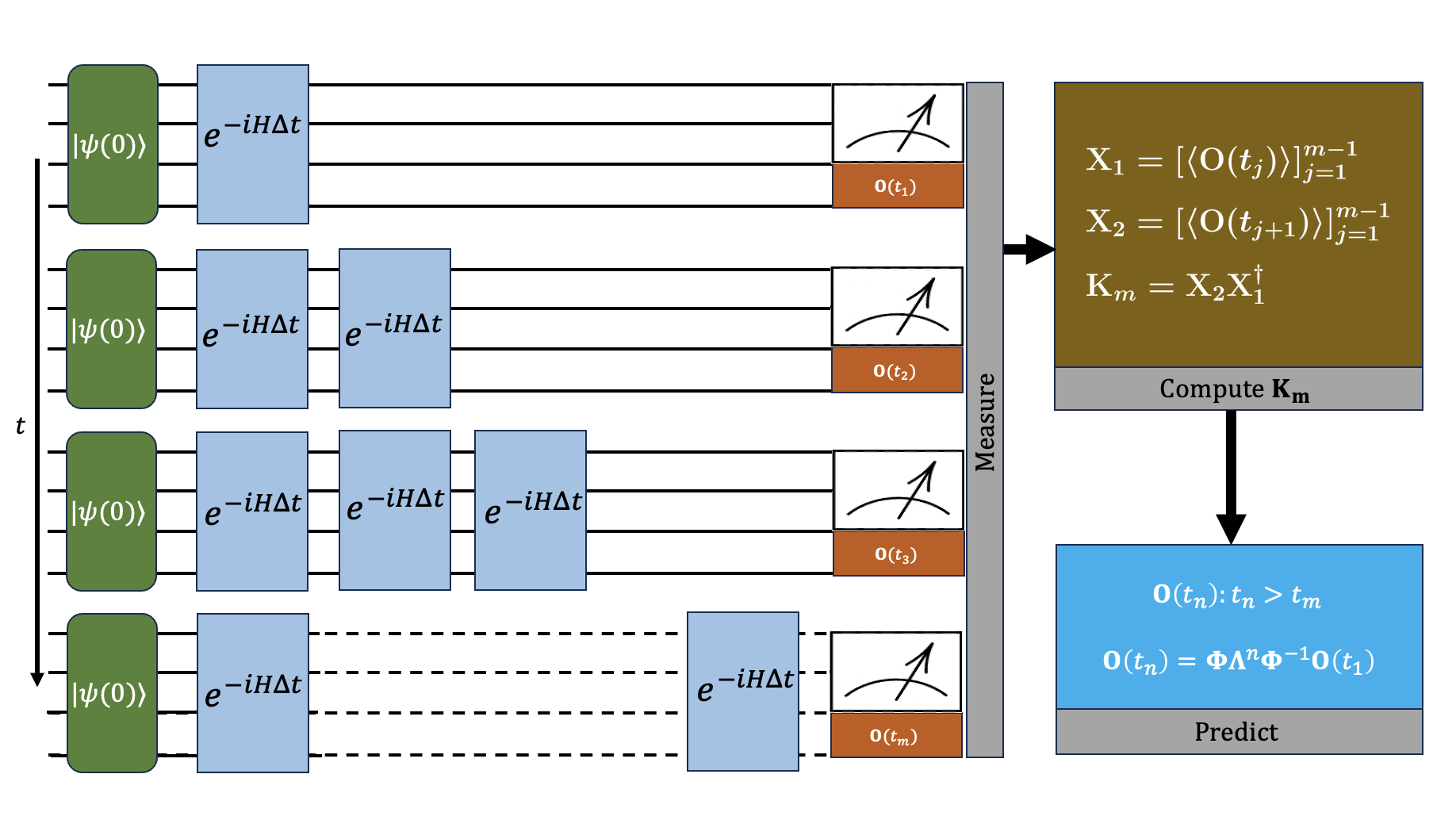} 
    \caption{\textbf{ Flow chart showing steps for efficient hybrid Hamiltonian simulation using a quantum processor and DMD:} First measure the physical observables of interest  in a quantum processor during a small time window, and then construct a reduced order linear model from these data to approximate the actual dynamics by using DMD. We use the constructed model to extrapolate long time dynamics of the physical observables. Evaluate the relevant modes using DMD. Estimate observables for long time simulation.}
  
    \label{fig:flowchart}
\end{figure*}

\section{Method}
\subsection{Dynamic mode decomposition}
We focus on simulating the dynamics of a spatially varying observable $\tO(\tx,t)\equiv \qty[ O(x_{1},t), O(x_{2},t), ...,O(x_{N},t) ]^{T}$ subjected to a time-independent many-body system Hamiltonian $H$. The equation of motion for such an observable is
 \be 
 \dv{\ev{\tO}}{t}= i \mqty[ \ev{\qty[H,O_1]}\\ \vdots \\ \ev{\qty[H,O_N]}].
 \label{eq:hberg}
 \ee 
In the above equation, we define $O_{j}(t)=O(x_j,t) \in \mathbb{C}^{N_H \times N_H}$, and $\ev{O} = \expval{O}{\psi}$ holds for any $N_H\times N_H$ matrix $O$. Solving equation~\eqref{eq:hberg} is challenging in practice because the analytical form of it is usually unknown. Even if we approximate it through many-body perturbation theory, the equation may be too complex to solve efficiently. In this work, we evolve the quantum state $\ket{\psi(t)}$ to measure $\ev{O(\tx,t)}$ up to $m$ time steps, and then use these data to construct a reduced order model for the dynamics of $\tO$. We use the equation-free data-driven dynamic mode decomposition (DMD) to build this model. The theoretical foundation of DMD lies in the Koopman operator theory which is a widely recognized approach that explains how observables change over time within a dynamic system ~\cite{DMDtoKoop,Koopman1,Koopman2}. The Koopman's operator represents the dynamics of the system as a linear transformation on a space of observable functions. By studying it, we can analyze the behavior and predict the future states of the system.
According to the Koopman operator theory~\cite{DMDtoKoop}, for any $\Delta t>0$, there exists a linear, infinite-dimensional operator $\mathcal{K}_{\Delta t}$, such that
\be 
\ev{\tO (t+\Delta t)} = \mathcal{K}_{\Delta t} \ev{\tO (t)}, 
\label{eq:k_eq}
\ee 
holds.


The aim of DMD is to approximate the infinite-dimensional operator $\mathcal{K}_{\Delta t}$ with a finite-dimensional linear operator which can be written in a matrix form $\tK$. 
Within the DMD framework, we measure the observable $\ev{O(\tx,t)}$ at $m$ time steps and record the resulting temporal `snapshots' are recorded in two matrices,
\be
\begin{aligned}
\tX_1 &= \mqty[| & | & \; & \; & | \\ \ev{{\tO}^1}& \ev{{\tO}^2}& . & . &\ev{{\tO}^{m-1}} \\ 
| & | & \; & \; & | ], \\
\tX_2 &= \mqty[| & | & \; & \; & | \\ \ev{{\tO}^{2}}& \ev{{\tO}^{3}}& . & . &\ev{{\tO}^{m}} \\ 
| & | & \; & \; & | ].
\end{aligned}
\ee
with ${\tO}^j=\tO(t_j)$, where we take the time steps $t_n = (n-1)\Delta t,\; n = 1,..,M+1$.


In evaluating $\ev{\tO(t_j)}$, the time-evolved state $\ket{\psi(t)} = e^{-iHt}\ket{\psi(0)}$ is prepared using the Trotter decomposition of the evolution operator $e^{-iHt}$. We will discuss more details about the Trotter method in a later subsection.  

From these data matrices, we can find an $N\times N$ matrix $\tK_m$ which approximates $\mathcal{K}_{\Delta t}$ in the sense that $\tK_m = \bPhi_m \bLambda_m  \bPhi_m^\dagger$ with its eigenvalues ($\bLambda_m$) and eigenvectors ($\bPhi_m$) close to $\mathcal{K}_{\Delta t}$'s.  The finite-dimensional approximation to the Koopman operator $\mathcal{K}_{\Delta t}$ is obtained by solving the following linear least squares problem
\begin{equation}
 \min_{\tK_m} \| \tK_m \mathbf{X}_1 - \mathbf{X}_2 \|_F^2,
\label{eq:lsq1}
\end{equation}
with the solution
\begin{equation}\label{eq:A1}
    \tK_m = \mathbf{X}_2{\mathbf{X}_1^\dagger},
\end{equation}
where ${\mathbf{X}_1^\dagger}$ is the Moore-Penrose pseudoinverse of $\mathbf{X}_1$ that can be computed from the singular value decomposition (SVD)~\cite{SVD} of $\mathbf{X}_1$. In practice, if $\mathbf{X}_1$ has low-rank structure, we will take a truncated SVD instead of full SVD in order to enhance efficiency. Then after obtaining an $r\times r$ matrix $\widetilde{\tK}_m$ with $r$ the truncated rank, we need to project it back to the original $N$-dimensional observable space.

Once we get $\mathbf{X}_1^\dagger$, and hence $\tK_m$ is obtained,  from the eigendecomposition $\tK_m\bPhi_m = \bPhi_m\bLambda_m$, we can construct the solution of~\eqref{eq:hberg} at a later time step $n>m$ as,
\be
\ev{\tO(t_{n})} = \bPhi_m \bLambda_m^{n-1} \bPhi_m^{\dagger}\ev{\tO(t_{1})}.
\ee
We can also get the approximation at any time $t>0$ by taking $\mathbf{\Omega}_m = {\rm{ln}}\bLambda_m/\Delta t$, and then we get
\be
\ev{\tO(t)} = \bPhi_m\exp(\mathbf{\Omega}_mt)\bPhi_m^{\dagger}\ev{\tO(t_{1})}.
\ee
Since ${\mathbf{\Omega}}_m$ is a diagonal matrix, the real part of each diagonal element describes the changes in amplitude, while the imaginary part describes the frequencies of the oscillation in the approximated dynamics.

The above described procedure is shown in a flowchart in Fig.~\ref{fig:flowchart} and the algorithm is summarized in Algorithm~\ref{alg:DMD}.

In short, our hybrid method includes time evolution of a quantum state in a quantum computer using Trotterization and simultaneous measurement of a set of observables. Using the measured data we utilize the DMD framework to predict long time dynamics of the observable. Both Trotter based methods and DMD are popular methods in their respective communities of quantum algorithms ~\cite{wiebe2011simulating,Pastori2024,babbush2015chemical} and differential equation solvers{~\cite{DMD0,schmid2011application, schmid2011applications, schmid2012decomposition, seena2011dynamic}, }respectively. 
Our key contributions are applying the designed hybrid method to simulate the quantum many-body dynamics from different quenched model systems accurately and establishing a predictive error bound of DMD for our application. 
To facilitate later discussions, we summarize the main points of Trotter based time evolution in the next subsection.
\begin{algorithm}[H]
\caption{The DMD procedure}\label{alg:DMD}
\begin{algorithmic}[1]
\State $\mathbf{X}_1 \gets \left[\ev{{\tO}^1}, \, ..., \, \ev{{\tO}^{m-1}}\right]$, $\mathbf{X}_2 \gets \left[\ev{{\tO}^{2}}, \, ..., \, \ev{{\tO}^{m}}\right]$
\State Truncated SVD of $\mathbf{X}_1$ with rank $r$: $\mathbf{X}_1 = \widetilde{\mathbf{U}}\widetilde{\mathbf{\Sigma}}\widetilde{\mathbf{V}}^T$
\State $\widetilde{\tK}_m \gets \widetilde{\mathbf{U}} \mathbf{X}_2 \widetilde{\mathbf{V}} \widetilde{\mathbf{\Sigma}}^{-1}$ \Comment{The projected Koopman operator}
\State Eigen-decomposition of $\widetilde{\tK}_m$: $\widetilde{\tK}_m\mathbf{W}_m = \mathbf{W}_m{\mathbf{\Lambda}_m}$
\State $\mathbf{\Phi}_m \gets \mathbf{X}_2\widetilde{\mathbf{V}}\widetilde{\mathbf{\Sigma}}^{-1}\mathbf{W}_m$ \Comment{DMD modes}
\State $\mathbf{\Omega}_m \gets {\ln{\mathbf{\Lambda}_m}}/{\Delta t}$, $\mathbf{b} = \mathbf{\Phi}_m^\dagger \ev{{\tO}^1}$ 
\State $\ev{\tO(t)}\gets \mathbf{\Phi}_m\exp(\mathbf{\Omega}_mt)\mathbf{b}$ \Comment{Reconstruction and prediction}
\end{algorithmic}
\end{algorithm}

\subsection{Trotter Suzuki method} 
Given a Hamiltonian of the form $H = \sum_{j=1}^{J} H_j$, we want to simulate $e^{-iHT}$ by a sequence of exponentials like $e^{-iH_{j}\Delta t}$, where the total simulation time $T$ is divided into $M$ small steps of width $\Delta t$ such that $M = T/\Delta t$.  
Time evolution of the state $\ket{\psi(T)} = e^{-iHT}\ket{\psi(0)}$ is prepared using the Trotter decomposition. Clearly the resource required (which can be estimated by the number of CNOT gates) increases with $t$, and therefore the time of evolution is highly limited by the available resources in the quantum hardware.  
For a time evolution of time $t$ the unitary $U = e^{-iHT}$ is approximated by a unitary $\Tilde{U}$ using the product formula,
\be 
\Tilde{U}(t) = \prod_{k=1}^{M}\prod_{j=1}^{J}e^{-iH_{j}\Delta t}
\ee


\section{Analysis of Predictive Accuracy}
Error analysis  for our method has two directions. First, how the error scales w.r.t time by keeping $m$, the number of snapshots we use in DMD, fixed. Second, scaling of the errors as a function of $m$ at a fixed future time. We will derive an analytical form for the first case, and provide numerical evidence for the latter. Our derivation of error upper bounds by using DMD follows as a special case of the predictive accuracy analyzed in~\cite{lu2020prediction}. The original work was done for a differential equation of the form 
\be 
\dv{\tu}{t}=\mathcal{A}\tu + \tf(\tu)
\ee 
where $\mathcal{A}$ is a linear operator, $\tu=[u_1, ..., u_N]^T\in\mathbb{C}^N$ is a vector of observables and $\tf$ is an external source/sink  term. In terms of Koopman's operator the above equation in discretized time with a small $\Delta t$ can be written as,
\be \label{eq:pde}
\tu^{n+1} \equiv \mathcal{K}_{\Delta t} \tu^{n} \approx \mathcal{A}\tu^{n} + \Delta t\, \tf^{n}, \; n=1, ..., M,
\ee 
where $\tf^{n}$ is the $n$-th discretized snapshot of $\tf$.
Calling the approximate eigenvalues and the eigenvectors of $\mathcal{K}_{\Delta t}$ as $\bLambda_m$ and $\bPhi_m$  let us first posit the following \cite{drmac2018data}, 
\begin{lemma}\label{lemma:e_m}
Define the global truncation error on the $m$-th temporal snapshot by $\epsilon^{m}=\normalfont{\tu}^{m} - \normalfont{\tu}^{m}_{\rm{DMD}}$, then DMD is designed such that 
\be\nonumber
\|\epsilon^{m}\|_2 := \sqrt{\sum_{j=1}^N|\epsilon_j^m|^2}, \; \epsilon_j^m = u_j^m-u_{j,\rm{DMD}}^m, \; j=1, ..., N
\ee
is minimum.
\end{lemma}
The error at the $n$-th time step ($n>m$) is then given by the following theorem \cite{lu2020prediction}
\begin{theorem} 
 \label{theorem:local_error}
    Define the local truncation error 
     \be 
     {\tau^{n}} = \normalfont{\tu}^{n+1} - {\mathbf{K}_m\normalfont{\tu^{n}}}
     \ee 
     {with $\mathbf{K}_m\normalfont{\tu^{n}} = \bPhi_m\bLambda_m\bPhi_m^{-1}\normalfont{\tu^{n}}$.}
     Then for any $n\ge m$, $\norm{\tau^{n}}_2 \le \epsilon_m$
where
\be
\begin{aligned}\label{eq:local_error_bound}
\normalfont
\epsilon_m =& c_m\Bigg[ \norm{\normalfont{\tu}^{0}}^2_2 +\Delta t  \sum_{k=1}^{n-1}\max \qty{ \norm{\normalfont{\tf}^{k} }^2_2, \norm{\normalfont{\tf}^{k+1}}^2_2 }  \Bigg]^{\frac{1}{2}},
\end{aligned}
\ee
with $\norm{\normalfont{\tf}^{k} }^2_2 := \sum_{j=1}^N|{\rm{f}}^k_j|^2$.
\end{theorem}
In the equation, $c_m$ is a constant depending on the number of snapshots $m$ used in DMD and satisfies $\norm{\mathcal{K}_{\Delta t}- \bPhi_m\bLambda_m\bPhi_m^{-1}}_{2}\le c_m$. Here the $2$-norm on the left hand side is defined by
\be
\begin{aligned}\nonumber
\max_{\substack{\phi(t)\in\mathbb{C}^N, \phi(t)\neq{\mathbf{0}},\\\tilde{t}\geq 0}}&\frac{\left\|\left(\mathcal{K}_{\Delta t}\phi(t)\right)|_{t=\tilde{t}}- \bPhi_m\bLambda_m\bPhi_m^{-1}\phi(\tilde{t})\right\|_2}{\|\phi(\tilde{t})\|_2}
\end{aligned}
\ee
where we use a shorthand  notation $\left(\mathcal{K}_{\Delta t}\phi(t)\right)|_{t=\tilde{t}}$ to denote 
\be\nonumber
\begin{aligned}
\begin{bmatrix}
\mathcal{K}_{\Delta t}\phi_1(t)\\\vdots \\
\mathcal{K}_{\Delta t}\phi_N(t)
\end{bmatrix}_{t=\tilde{t}} = \begin{bmatrix}
\phi_1(\tilde{t}+\Delta t)\\\vdots \\
\phi_N(\tilde{t}+\Delta t)
\end{bmatrix} = \phi(\tilde{t}+\Delta t).
\end{aligned}
\ee
The $c_m$ term can be viewed as a measure of how closely we can approximate the Koopman operator in an $N$-dimensional subspace.

{In this and subsequent theorems, we adopt the assumption that $\mathbf{u}^n$ represents the exact snapshot at $t_n$, i.e., $\mathbf{u}^n=\mathbf{u}(t_n)$ for simplicity. But in practice, errors may inevitably arise in the data, stemming from factors such as data measurement or numerical simulation as in our case. This will result in perturbed matrices $\bPhi_m$ and $\bLambda_m$, which are denoted by $\hat{\bPhi}_m$ and $\hat{\bLambda}_m$, respectively. However, as long as the data matrices $\tX_1$, $\tX_2$ are well-conditioned, the product $\hat{\bPhi}_m\hat{\bLambda}_m\hat{\bPhi}_m^{-1}$ will still give a reasonable approximation to the Koopman operator $\mathcal{K}_{\Delta t}$. Specifically, we will have $\norm{\mathcal{K}_{\Delta t}- \bPhi_m\bLambda_m\bPhi_m^{-1}}_{2}\leq\hat{c}_m$, where $\hat{c}_m$ is a perturbed yet still small constant.}

\begin{theorem}\label{theorem:error_b}
The global error at the n-th time step $\epsilon^n$ by using $m$ snapshot data in DMD is bounded above by
\be\label{eq:global_error_bound}
\norm{\epsilon^{n}}_{2} \le \norm{\bPhi_m}_{2} \norm{\bPhi_m^{-1}}_{2}\qty[\norm{\epsilon^{m}}_{2} + (n-m)\epsilon_m ]
\ee 
with $\epsilon_m$ defined by~(\ref{eq:local_error_bound}), if the eigenvalues of $\mathcal{A}:\{\lambda_{1}, \lambda_{2},..,\lambda_{N}\}$ satisfies $\max_{1\le k \le N} \abs{\lambda_k}\le 1.$
\end{theorem}

The proof of {\textbf{Theorem}}~\ref{theorem:error_b} can be found in~\cite{lu2020prediction}. This theorem states that the global error mainly comes from two parts: 1) the approximation to the Koopman operator; 2) the intrinsic property of the dynamics. In the first part, the $c_m$ factor in~(\ref{eq:local_error_bound}) measures how well DMD approximates the Koopman operator. Since the number of DMD modes $r$ is determined by the rank of $\tK_m$, when $N$ is small, i.e. with insufficient spatial resolution, the DMD modes does not fully represent the spectral property of the Koopman operator. Thus one way to reduce $c_m$ is by improving the spatial resolution. In the second part, the potentially large condition number $\|\bPhi_m\|_2\|\bPhi_m^{-1}\|_2$ is caused by the departure from normality of the projected Koopman operator $\tK_m$, i.e., the eigenvectors of the $\tilde{\tK}_m$ are far from orthogonal. 

We will now provide local and global truncation error bounds using the above two theorems for our case of quantum dynamics simulation. 
\begin{theorem} \label{theorem:local_error_quantum}
    The local and the global truncation error by using DMD for the dynamical equation~(\ref{eq:hberg}) is respectively governed by,
    \bea 
    \epsilon_m &\leq& c_m\norm{\normalfont{\tO}}_F\left( 1 +  2n \Delta t  \norm{H}^2_F \right)^{\frac{1}{2}},\nonumber \\
\norm{\epsilon^{n}}_2 &\le& \norm{\bPhi_m}_2\norm{\bPhi_m^{-1}}_2 \Big[\Delta_m \notag\\
&& {}+ (n-m) c_m\norm{\normalfont{\tO}}_F\left( 1 +  2n \Delta t  \norm{H}^2_F \right)^{\frac{1}{2}} \Big] \notag,
\eea 
    where $\Delta_m > 0$ is a constant for a fixed $m$. As a result, the global truncation error scales as $O(t^{3/2})$.
\end{theorem}

\begin{proof}
 By comparing~(\ref{eq:hberg}) and~(\ref{eq:pde}), we notice that $\tu^{n} = \ev{\tO(t_n)}$,  $\mathcal{A}$ is an $N$-dimensional identity matrix and 
 \be 
 \tf^{n} = i\mqty[ \ev{\qty[H,O_1(t_n)]}\\ \ev{\qty[H,O_2(t_n)]}\\. \\. \\ \ev{\qty[H,O_N(t_n)]}] 
 \ee 
 Using Lemma~\ref{lemma:util} in the Appendix we can write 
 \bea 
\norm{\normalfont{\tf}^{n}}^2_2 &=& \sum_{j=1}^{N} \left|\ev{\qty[ H, O_j(t_n) ]}\right|^2\notag \\
&\le& \sum_{j=1}^{N}\norm{\qty[ H, O_j(t_n) ]}^2_F 
\le 2\sum_{j=1}^{N}\norm{H}^2_F\norm{O_j(t_n)}^2_F \notag \\
&\le& 2\norm{H}^2_F\sum_{j=1}^{N}\norm{O_j}^2_F 
= 2\norm{H}^2_F\norm{\tO}_F^2
\label{eq:fn}
 \eea 
where we take $\sum_{j=1}^{N}\norm{O_j}^2_F \equiv \norm{\tO}_F^2$. Note that $\norm{\tO}_F^2$ does not change with time. To see this, consider the simplification of the term $\norm{O_j(t_k)}^2_F$: $\norm{O_j(t_k)}^2_F = \norm{e^{iHt_k}O_je^{-iHt_k}}^2_F = \norm{O_j}^2_F$ since $e^{\pm iHt_k}$ are unitary matrices. In the above derivation, we also used the result from Ref.~\cite{bottcher2008frobenius}: $\norm{\qty[A,B]}^2_F\le 2\norm{A}^2_F\norm{B}^2_F$.

Now by putting~(\ref{eq:fn}) in~(\ref{eq:local_error_bound}), we obtain the local truncation error 
\bea
\epsilon_m &=& c_m\left( \sum_{j=1}^{N} \norm{O_j(t_0)}^2_F \right. \notag \\
&& \left. {}+ \Delta t \sum_{k=1}^{n} \max \qty{ \norm{\normalfont{\tf^{(k)}}}^2_2, \norm{\normalfont{\tf^{(k+1)}}}^2_2} \right)^{\frac{1}{2}}  \notag \\
&\le& c_m\left( \sum_{j=1}^{N} \norm{O_j}^2_F +  2\Delta t \sum_{k=1}^{n} \norm{H}^2_F\norm{\tO}_F^2 \right)^{\frac{1}{2}} \notag \\
&=& c_m\norm{\tO}_F\left( 1 +  2n \Delta t  \norm{H}^2_F \right)^{\frac{1}{2}}
\eea 
The global truncation error can be obtained by substituting $\epsilon_m$ in~(\ref{eq:global_error_bound}) and taking $\norm{\epsilon^{m}}_2$ as a constant $\Delta_m$ (by Lemma.~\ref{lemma:e_m}). Thus we prove Theorem.~\ref{theorem:local_error_quantum}:
\bea 
\norm{\epsilon^{n}}_2 &\leq& \norm{\bPhi_m}_2\norm{\bPhi_m^{-1}}_2 \left[\Delta_m \right. \notag \\
&& \left. {}+ (n-m) c_m\norm{\normalfont{\tO}}_F\left( 1 +  2n \Delta t  \norm{H}^2_F \right)^{\frac{1}{2}} \right] \notag
\eea 
Now considering predictions for long time $n\gg m$, we can write $(n-m)\approx t/\Delta t$. Using $n\Delta t = t$, we can rewrite the upper bound as,
\bea
\norm{\epsilon^{n}}_2 &=& \norm{\bPhi_m}_2\norm{\bPhi_m^{-1}}_2 \left[\Delta_m \right. \notag \\
&& \left. {}+ \frac{t}{\Delta t} c_m\norm{\normalfont{\tO}}_F\left( 1 +  2t  \norm{H}^2_F \right)^{\frac{1}{2}} \right] \notag
\eea 
The above result shows that the scaling of error is at most $O(t^{3/2})$, which is suboptimal.

\end{proof}

For the second part of error analysis, our focus shifts to examining the relationship between the global truncation error at time $T=t_{M+1}$ and the number of snapshots $m$ used to construct DMD. We will only provide numerical evidence as part of this error analysis.  We notice that the error decreases initially with respect to $m$ before it levels off and oscillates around a constant value. This is most likely due to the fact that the observables of interest here do not span an invariant subspace of the Koopman operator. As a result, the approximation error $c_m$ in~(\ref{eq:local_error_bound}) will not decrease to 0. As we mentioned before, one way to reduce $c_m$ is to increase the spatial resolution of the observables. We will examine the numerical results shown in  Figure~\ref{fig:err_plot}(c \& d) in the next section to check the validity of this approach. The error estimate of the Trotter method can be found in \cite{berry2007efficient,childs2021theory}.



\section{Results}

\subsection{Hubbard Model and quenching} 
To demonstrate our method, we choose the quenching process in a finite size Hubbard Hamiltonian. We start from the ground state $\ket{\psi_0}$ of a non-interacting Hamiltonian $H_0$ on $L$ sites. 
\be
H_0 = -\sum_{i,j,\sigma} \tau_0\qty(c_{i,\sigma}^{\dagger}c_{j,\sigma} + h.c.) 
\ee
We time evolve the quantum state w.r.t. an interacting Hamiltonian $H_1$,
\be
\begin{aligned}
H_1 =& -\sum_{i,j,\sigma} \tau_1\qty(c_{i,\sigma}^{\dagger}c_{j,\sigma} + h.c.) + U\sum_{j}n_{j \uparrow}n_{j \downarrow} - \mu \sum_{j,\sigma}n_{j,\sigma}
\end{aligned}
\ee
We choose $\tau_0 = 1.0$ and $\tau_1=0.1$ specifically in our example. To preserve particle-hole symmetry, we choose $\mu = U/2$. In our calculation, the parameters $U$ and $t$ are in the units of $\tau_0$.
The Hamiltonian is mapped to qubits using Jordan Wigner transformation. Throughout the rest of the paper, we consider Hubbard model at half-filling with total spin and its `z' component ($S_z$) to be zero, i.e, the number of electrons is the same as the number of lattice sites $L$, $N_{\uparrow} = N_{\downarrow}$ and we assume open boundary condition. 

By quenching a non-interacting metallic state through the introduction of a Hubbard interaction, one can observe a suppression of coherent tunneling dynamics.
The period within the dynamics is determined primarily by the interaction strength between the electrons. A large number of theoretical works using analytical and classical numerical methods exist seeking to characterize and
understand nonequilibrium dynamics in quantum systems~\cite{cazalilla2010focus,dziarmaga2010dynamics, polkovnikov2011colloquium}. Such works have been motivated by experiments like ultracold-atom where far-from-equilibrium dynamics can be observed. Other physical phenomena of scientific interest include the  collapse and revival of matter waves with Bose-Einstein
condensates in optical lattices~\cite{greiner2002collapse, will2013coherent}, coherent quench dynamics of a Fermi sea in a Fermi-Bose mixture~\cite{will2015observation}, non-thermal behavior in near-integrable experimental regimes~\cite{will2015observation, gring2012relaxation}, and equilibration in Bose-Hubbard-like systems~\cite{ trotzky2012probing}.

Our simulation proceeds by the following steps. First, we create the initial state,  then we choose the set of density matrix operators $\rho_{pq} = \ev{c_{p,\sigma}^{\dagger}c_{q,\sigma}}\; : (p,q) \in \qty{1,2,..,L}$ as our set of observables $\tO$. The lattice site indices are our position coordinates $x$ and we want to observe the dynamics of $\tO$ with time $t$. We then use Trotter method to simulate the Hamiltonian evolution upto time $t_m$ and measure the correlation function $\tO$ at every time step. After that, we use DMD on the measure set of data to predict long time values of the set $\qty{ \ev{c_{p,\sigma}^{\dagger}c_{q,\sigma}} }$.  Density matrix elements can be computed for different pairs of sites within a lattice. For the reason of compactness, we present our results in momentum space using a linear combination of all pairs in real-space 
\be  
n_{\tk,\sigma} = \frac{1}{L}\sum_{p,q} \ev{c_{p,\sigma}^{\dagger}c_{q,\sigma}} e^{-i\tk (p-q)}, \label{eq:rho_momentum}
\ee  
where $\tk$ is the momentum index.

\begin{figure}[ht!]
    \centering
    \includegraphics[width=0.9\linewidth]{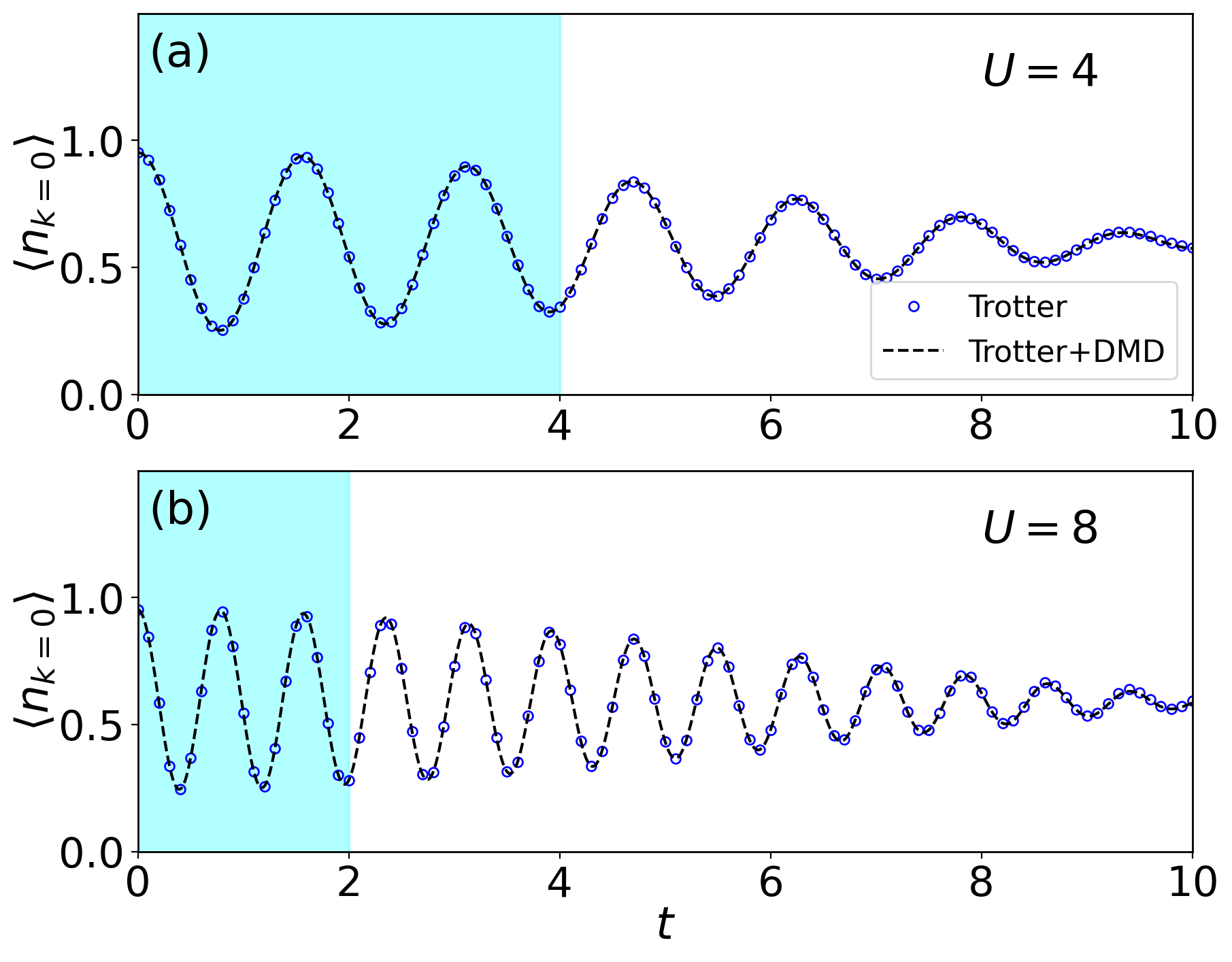} 
    \caption{\textbf{Extrapolation result for $\boldsymbol{n_{k=0}}$:} The shaded area represents the data used for DMD $L=6$ in this example. (a) $U=4.0$, we use the first 400 snapshots for the extrapolation, and iHODMD$(10,4,4)$ is used. (b) $U=8.0$, we use the first 200 snapshots for the extrapolation, and iHODMD$(20, 2, 2)$ is used.}
    \label{fig:nk_time}
\end{figure}

\subsection{Technical details}
In our system, since the scales of different entries in the density matrix vary a lot, we need to preprocess the data for $j=1, ..., L^2$ before constructing the data matrices by
\be
\langle \tilde{O_j}(t)\rangle = \dfrac{\langle O_j (t) \rangle - {\rm{mean}}{(\langle O_j(t_1:t_n)\rangle )}}{{\rm{std}}(\langle O_j(t_1:t_n)\rangle)}.
\ee
DMD becomes exact if the spatial discretization is infinite. Working with finite size lattices suffer from the problem of insufficient spatial degrees of freedom. Additionally, model Hamiltonians like Hubbard model have spatial and other symmetries that make many of the discretized observables trivial. Such observables can be predicted easily from symmetry arguments or can be ignored for non-trivial dynamics simulation.  To compensate for the resulting insufficient spatial resolution, we introduce an improved higher order DMD (iHODMD) derived from the time-delay embedding theory~\cite{broomhead1986extracting,packard1980geometry,Pan2020,Taken}. 
iHODMD is similar to Algorithm~\ref{alg:DMD}, but we construct the data matrices for each observable $\langle O_j(t)\rangle$, $j\in\{1, ..., L^2\}$ with
\begin{equation*}
\begin{aligned}
	&\mathbf{\tilde{X}}_1 = \begin{bmatrix}
    		\langle \tilde{O_j}^{1}\rangle & \langle \tilde{O_j}^{n_g+1}\rangle & \cdots & \langle \tilde{O_j}^{(n_l-1)\times n_g+1}\rangle \\
    		\langle \tilde{O_j}^{2}\rangle & \langle \tilde{O_j}^{n_g+2}\rangle & \cdots & \langle \tilde{O_j}^{(n_l-1)\times n_g+2}\rangle \\
    		\vdots & \vdots & \vdots & \vdots\\
    		\langle \tilde{O_j}^{n_s}\rangle & \langle \tilde{O_j}^{n_g+n_s}\rangle &\cdots & \langle \tilde{O_j}^{(n_l-1)\times n_g+n_s}\rangle
    	\end{bmatrix}, \\
	&\mathbf{\tilde{X}}_2 =\begin{bmatrix}
        	\langle \tilde{O_j}^{1+\tau}\rangle & \langle \tilde{O_j}^{n_g+1+\tau}\rangle & \cdots & \langle \tilde{O_j}^{(n_l-1)\times n_g+1+\tau}\rangle \\
        	\langle \tilde{O_j}^{2+\tau}\rangle & \langle \tilde{O_j}^{n_g+2+\tau}\rangle & \cdots & \langle \tilde{O_j}^{(n_l-1)\times n_g+2+\tau}\rangle \\
        	\vdots & \vdots & \vdots & \vdots\\
        	\langle \tilde{O_j}^{(n_s+\tau)}\rangle & \langle \tilde{O_j}^{n_g+n_s+\tau}\rangle &\cdots & \langle \tilde{O_j}^{(n_l-1)\times n_g+n_s+\tau}\rangle
        \end{bmatrix},
\end{aligned}
\end{equation*}
where $n_s$ is the order of iHODMD, $n_g$ denotes the time difference of two adjacent columns in the data matrices, and $\tau$ denotes the time difference between two data matrices. We call the iHODMD with the above data matrices iHODMD($n_s$,$n_g$,$\tau$). The number of columns $n_l$ is chosen by $n_l = {\rm{floor}}\left((m-n_s-\tau)/n_g\right)+1$. This iHODMD method improves the original higher order DMD method presented in~\cite{HODMD} by introducing two additional parameters $n_g$ and $\tau$, which increases the flexibility of the algorithm. Since (\ref{eq:hberg}) holds for each observable $O_j$ independently, a similar error analysis still applies for iHODMD. We only need to take $\mathbf{O}(t)$ in the analysis by 
\be
\mathbf{O}(t) = \left[\tilde{O_j}(t),\, \tilde{O_j}(t+\Delta t), \,..., \, \tilde{O_j}(t+n_s\Delta t)\right]^T,
\ee
and take the time step size to be $\tau\Delta t$.

The major computational cost of iHODMD($n_s$,$n_g$,$\tau$) is still in the SVD of $\mathbf{\tilde{X}}_1$, which is $\mathcal{O}(\min(n_s^2L^4n_l, n_sL^2n_l^2))$. The memory cost of the algorithm is $\mathcal{O}(n_sL^2n_l)$.

\subsection{Simulator results}
We consider the Hubbard system with $L=6$. 
We first present our results for $\tk=0$ in (\ref{eq:rho_momentum}) in Figure~\ref{fig:nk_time}. Panels (a) and (b) are the results for $U=4.0$ and $8.0$, respectively.   The blue squares represent the data due to Trotter evolution and the black curve is the DMD prediction. The prediction is calculated based on the time-series data collected within the the shaded-cyan region. 
Here we use iHODMD once for all the $6$ trajectories in the momentum space. From these figures, we can see that DMD gives good extrapolation for both $U$s. For a larger interaction $U=8$, we can even use a smaller number of snapshots to get results accurate enough. A recent study~\cite{kaneko2025forecasting} claims that for quantum many-body systems, DMD can extrapolate dynamics to 5-10 times the length of its sampling window. We verified that this claim holds in this example as well as in the XXZ spin model we study in \ref{appendix:xxz}. Therefore, our hybrid method can achieve accurate DMD predictions beyond the characteristic time $t^*$ where classical simulations become intractable, provided we sample the data up to a time $t>t^*/5$.

As mentioned earlier, due to finite size and spatial symmetries of our system, only a handful of $\rho_{pq}$ are independent. Therefore we choose $\rho_{13}$ to show the prediction error in Figure~\ref{fig:err_plot}. We take 200 snapshots up to $t=2$ in the extrapolation for both cases. Figure~\ref{fig:err_plot}(a) shows the difference $\Delta = \rho_{13} - \rho_{13}^{\rm{DMD}}$ as a function of $t$ for different $U$. We show that the error is upper-bounded by $\sim t^{\frac{3}{2}}$, shown in magenta. To further demonstrate our method, we have applied our method to the dynamics of $\ev{S^{z}_{j}(t)S^{z}_{j+1}(t)}$ for quenching in XXZ spin model \cite{Smith_2019}, where $S^{z}_{j}$ is the spin magnetic moment at $j$-th site in a $L$ spin system. The details about the Hamiltonian and the quench protocol is given in~\ref{appendix:xxz}. The DMD error for the XXZ model simulation is shown in Figure~\ref{fig:err_plot}(b), where $O(t) = \ev{S^{z}_{1}(t)S^{z}_{3}(t)}$ for $L=6$ and $O(t) = \ev{S^{z}_{1}(t)S^{z}_{6}(t)}$ for $L=12$. The errors are bounded above by $t^{3/2}$ following our derivation.  As the Trotter splitting gives us error at $O(10^{-2})$ with the time step size $\Delta t=0.01$, the extrapolation results with $m\geq 200$ are generally acceptable. The oscillation behavior of the error with respect to $m$ is due to the nature of data-driven methods as we have observed before~\cite{DMDdiag,DMDtwotime}.

\begin{figure}[h!]
    \centering
    \includegraphics[width=1.0\linewidth]{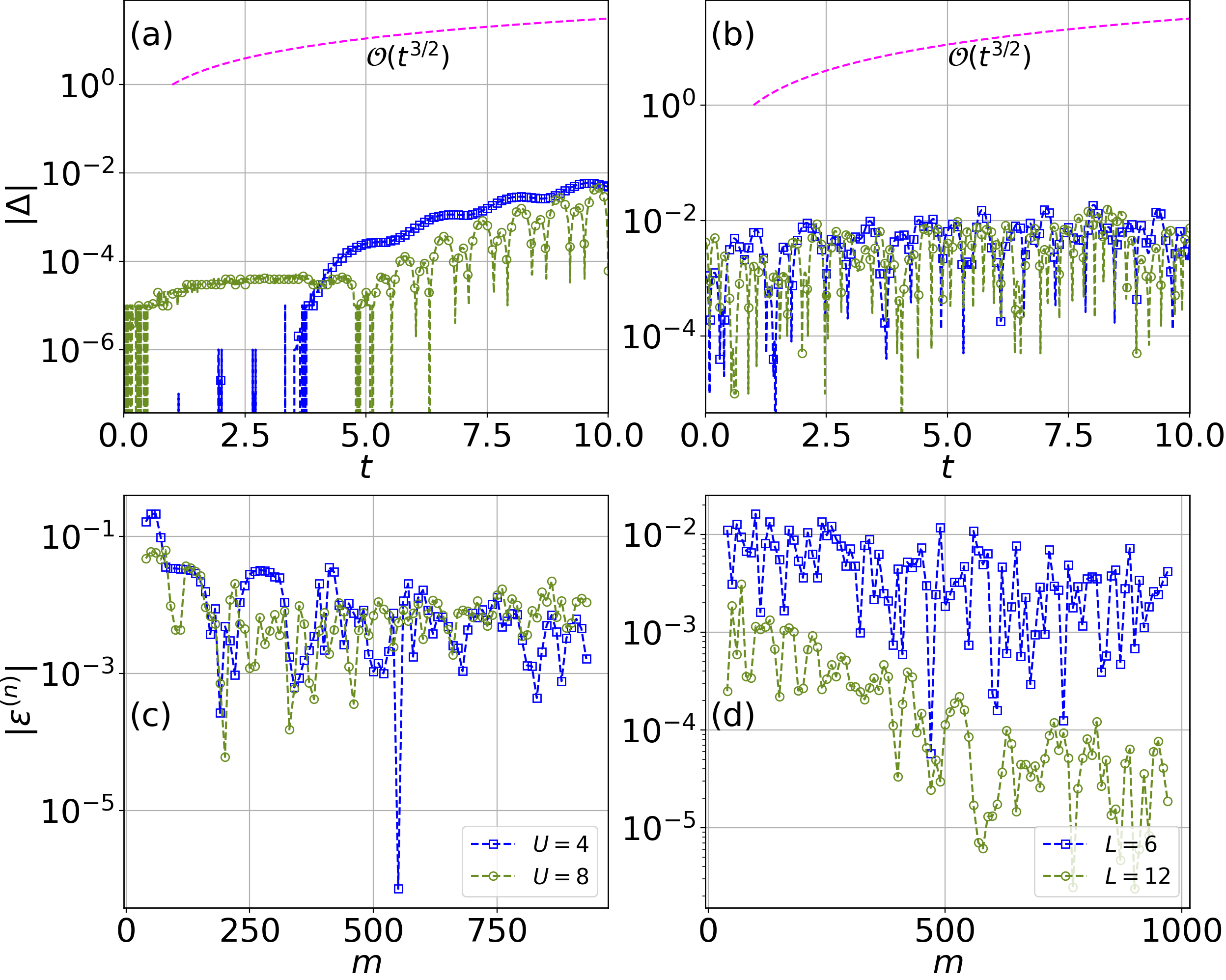} 
    \caption{\textbf{Errors between the original trajectories and the extrapolated results.}  Hubbard model (a) \& (c) XXZ model (b) \& (d). For (a), we use the first 400 (200) snapshots of $\rho_{13}$ in DMD extrapolation for $L=6$ (12 qubits), and $U=4.0$ or $8.0$ Hubbard model quenching, respectively. In DMD computation, we use iHODMD(20,10,10) for $U=4.0$ and iHODMD(20,2,2) for $U=8.0$. (b) DMD error for $\ev{S^{z}_{1}(t)S^{z}_{3}(t)}$ ($L=6$) and $S^{z}_{1}(t)S^{z}_{6}(t)$ ($L=12$) as a function of time for XXZ dynamics for $L=6$ and $L=12$. We use iHODMD(20,4,4) for both cases. (c) \& (d) shows the absolute error at the last time step $t=10$ for different number of snapshots $(m)$ taken in DMD. iHODMD parameters are the same as in (a) \& (c).} 
    \label{fig:err_plot}
\end{figure}

In Figure~\ref{fig:err_plot}(c \& d), we plot the  error at the final time step ($T=t_{M+1}$) as a function of number of snapshots taken in DMD for Hubbard and XXZ model, respectively. We see that the error decreases initially with an increasing $m$ before it reaches and oscillates around a constant value which is at the order of $\sim 10^{-3} - 10^{-5}$. We also notice that the error drops to a lower threshold as we go from $L=6$ to $L=12$ for the XXZ model (Figure~\ref{fig:err_plot}(d)). This is probably because having more observables in $\tO$ with $L=12$ improves the spatial resolution of the observables, which leads to a smaller $c_m$ in~(\ref{eq:local_error_bound}) and a better approximation to the Koopman's operator. We remark here that the performance of DMD can be further improved if an alternative set of observables that span an invariant subspace of the Koopman's operator are used~\cite{li2017extended,Matthew_edmd}. However, to the best of our knowledge, no systematic way has been established to determine the optimal choice of observables for DMD. Recent research suggested that combining DMD with a convolutional autoencoder for observable selection will lead to improved accuracy~\cite{lusch2018deep,haq2022dynamic}. We will examine this approach in future work.\\

\section{Conclusion}
We have successfully combined DMD, a well-known data-driven predictive method with quantum data  to predict long term observables in a quantum system. For a dynamical quantum system, the quantum data is collected in form of observable expectation values. Since the resources in terms of quantum devices (e.g, CNOTs) increase quickly as we simulate dynamics, long time simulation of large system sizes are very hard to run. By utilizing the predictive methods like DMD, we can run quantum simulations for a small number of time steps and predict future results. We have analytically shown that the error scales at most as $\mathcal{O}(t^{3/2})$.

Our method can be applied to any other dynamical systems as well and will find application in other forms of property prediction for quantum systems. We recently came across a work \cite{shen2023estimating} that uses a similar protocol for phase estimation. In summary, our method could be used in multiple dynamical problems to obtain meaningful results in the near term quantum devices.

\begin{acknowledgments}
The authors would like to thank Yizhi Shen and Nathan Wiebe for useful discussions. 
This work was supported by the ``Embedding QC into Many-body Frameworks for Strongly Correlated Molecular and Materials Systems'' project, which is funded by the U.S. Department of Energy, Office of Science, Office of Basic Energy Sciences (BES), the Division of Chemical Sciences, Geosciences, and Biosciences in the Center for Scalable and Predictive methods for Excitations and Correlated phenomena (SPEC) at Pacific Northwest National Laboratory under FWP 70942 (NG, WAdJ). This work was also supported by the Center for Computational Study of Excited-State Phenomena in Energy Materials (C2SEPEM) at the Lawrence Berkeley National Laboratory, which is funded by the U.\,S. Department of Energy, Office of Science, Basic Energy Sciences, Materials Sciences and Engineering Division, under Contract No. DE-AC02-05CH11231, as part of the Computational Materials Sciences Program (JY, CY). In addition, this work was supported by National Natural Science Foundation of China under Grant No. 12571422 (JY) and was sponsored by Shanghai Rising-Star Program (JY). 

\end{acknowledgments}

\appendix
\section{Operator Expectation}
\begin{lemma}\label{lemma:util}
For any unit wave function $|\psi\rangle\in\mathbb{C}^N$ and matrix $O\in\mathbb{C}^{N\times N}$,
\be
\langle\psi|O|\psi\rangle \leq \|O\|_2 \leq \|O\|_F,
\ee
where $\|\cdot\|_2$ and $\|\cdot\|_F$ stand for the $L^2$ and the Frobenius norms for matrices, respectively.
\end{lemma}
\begin{proof}
    From Cauchy inequality, we have
    \be
      \langle\psi|O|\psi\rangle \leq \langle\psi|\psi\rangle\langle O\psi|O\psi\rangle.
    \ee
    Recall the definition
    \be
      \|O\|_2 = {\rm sup}_{|\psi\rangle\neq 0}\dfrac{\langle O\psi|O\psi\rangle}{\langle\psi|\psi\rangle},
    \ee
    by using the fact that $\langle\psi|\psi\rangle=1$, we arrive at
    \be
      \langle\psi|O|\psi\rangle\leq \langle\psi|\psi\rangle\|O\|_2 = \|O\|_2\leq \|O\|_F,
    \ee
    where the last inequality holds naturally.
\end{proof}

\section{Quenching in XXZ model}
\begin{figure}[tbp]
    \centering
    \includegraphics[width=0.8\linewidth]{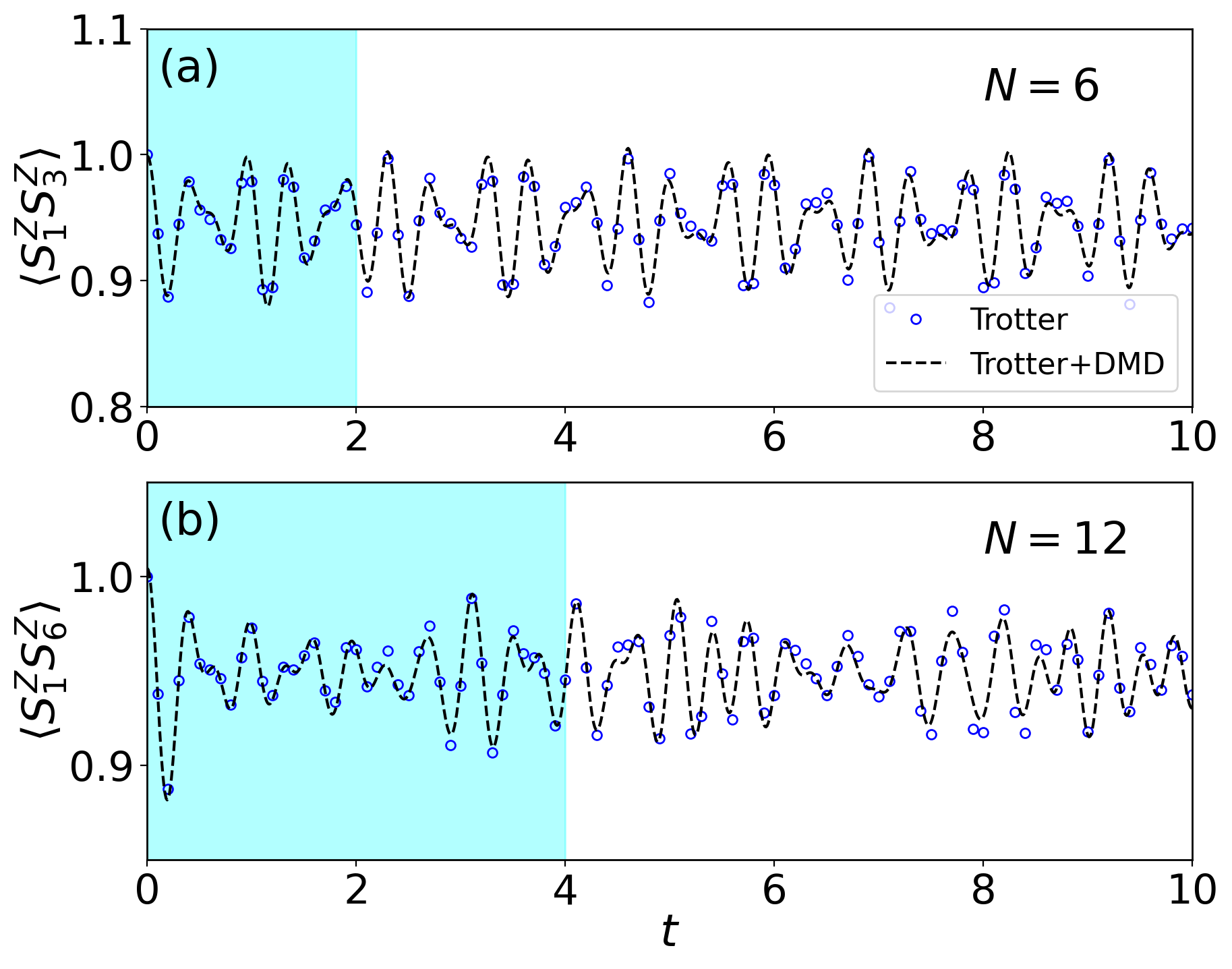} 
    \caption{\textbf{Extrapolation result for $\langle S_{j_1}^ZS_{j_2}^Z\rangle$:} The shaded area represents the data used for DMD with $L=6$ and $L=12$ (a) $L=6$, we use the first 200 snapshots for the extrapolation, and iHODMD$(20,1,1)$ is used. (b) $L=12$, we use the first 400 snapshots for the extrapolation, and iHODMD$(20,4,4)$ is used.}
    \label{fig:spin_time}
\end{figure}
We consider one-dimensional spin-1/2 chains consisting of $L$ spins, initially prepared in a domain wall configuration $\ket{...\uparrow\uparrow\downarrow\downarrow}...$. The time evolution after the quantum quench is led by the Hamiltonian,
\be
\begin{aligned}
H = -\sum_{j=1}^{L-1}\qty( X_{j}X_{j+1} + Y_{j}Y_{j+1}) + U \sum_{j=1}^{L-1}Z_{j}Z_{j+1} + h\sum_{j=1}^{L}Z_{j}
\end{aligned}
\ee
where $\qty{X,Y,Z}$ are Pauli matrices with eigenvalues $\pm 1$. In our calculation we choose $U=4.0$ and $h=0.1$. We calculate spin-spin correlation $\ev{Z_{j}Z_{j+1}}$ as a function of time. The results are shown in Figure~\ref{fig:spin_time}. The extrapolation has been done using the data within the shaded cyan region. For clarity in vision, we show the plots for $\ev{S^{z}_{1}(t)S^{z}_{3}(t)}$ for $L=6$ and $\ev{S^{z}_{1}(t)S^{z}_{6}(t)}$ for $L=12$. We have tested our methods for other sets of correlation functions and they yield similar accuracy.

\label{appendix:xxz}

\section{{Quenching in 10-site Hubbard model}}
{
\begin{figure}[tbh]
    \centering
    \includegraphics[width=0.8\linewidth]{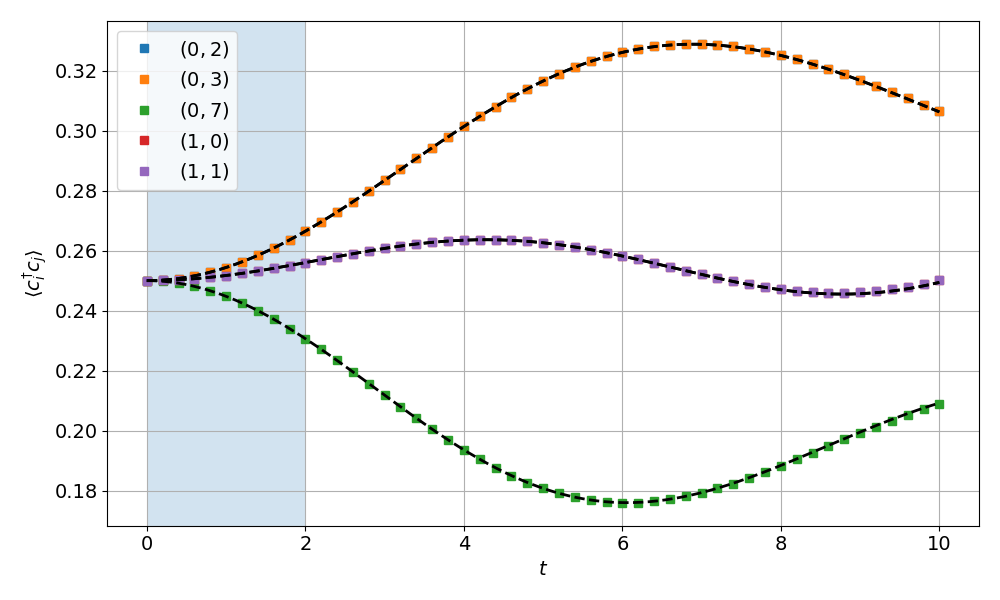} 
    \caption{\textbf{Extrapolation result for $\langle c_i^{\dagger} c_j\rangle$:} Results of DMD predictions for different $(i,j)$. The colored lines are predictions and black squares are exact results. The shaded area represents the data used for DMD with $L=10$ i.e, $20$ qubits and $U=0.1$. We have used the first 200 snapshots for the extrapolation, and iHODMD$(20,1,1)$ is used. }
    \label{fig:spin_time2}
\end{figure}
In Figure~\ref{fig:spin_time2}, we show the results for larger system size, $N_x =10$ using 20 qubits. We apply DMD to predict  $\langle c_i^{\dagger} c_j\rangle$ for half-filled Hubbard model. The DMD predictions are based on the first 200 time steps and we have compared the predictions up to 100 time steps. Different colored lines are for different sets of $(i,j)$. The exact values are plotted with black squares.  While we calculated all possible sets of correlation values, we show only the ones that vary significantly over time. The observables the remain constant has been eliminated from the plot for better visibility and clarity.  As evident from the plot, the results are in very good agreement with the the exact results. }

\label{appendix:hubbard}

\section{Advantages over the Trotter-based methods}
In this work, we only provide an error bound of DMD for quantum simulation. However, for completeness we also discuss the result for Trotter errors in simulating local variables and the gate complexity that can be avoided. Our method is aimed at combining Trotter based dynamics upto a certain time $t_m$ and then applying DMD. For comparison, we quote the result from \cite{childs2021theory}, Proposition 8 which states that if the entire dynamics was simulated using a $p$-th order Trotter formula, the error of simulating local observable will scale as $\mathcal{O}(\frac{t^{p+1}}{r^p})$, where $r$ is the number of Trotter steps. Using Trotter method will also require additional hardware resources. According to \cite{childs2021theory}, the gate count ($g_{\infty}$) for simulating local observables for a power law Hamiltonian with exponent $\alpha$ in $d$-dimension is $\mathcal{O}(t^{\frac{(\alpha(p+1)-d)(\alpha(dp+p+1)-(d+2)dp-d)}{p(\alpha-d)(\alpha+d^2-d(\alpha+2p+1)+\alpha p)}})$, which in the large $p$ limit simplifies to $t^{(1+d\frac{\alpha-1}{\alpha-2d})(1 + \frac{d}{\alpha - d})}$. By using DMD, the additional gate cost after $t>t_m$ can be completely ignored.  We have shown these two comparisons in Table \ref{table:1}. 
\begin{table}[!ht]
\centering
\caption{\textbf{Comparison of error and gate complexity}}
\label{table:1}
\begin{tabular}{||c | c | c ||} 
\hline
Method & Trotter/PF & QDMD \\  
\hline\hline
Error  & $\mathcal{O}(\frac{t^{p+1}}{r^{p}})$ &  $\mathcal{O}(t^\frac{3}{2})$ \\
\hline
$g_{\infty}$ & $\mathcal{O}(t^{\frac{(\alpha(p+1)-d)(\alpha(dp+p+1)-(d+2)dp-d)}{p(\alpha-d)(\alpha+d^2-d(\alpha+2p+1)+\alpha p)}})$ & N/A\\

\hline
\end{tabular}
\end{table}

\section{The relation between the global error and $\Delta t$}
The time step $\Delta t$ affects the global error in a two-fold way. First, $\Delta t$ controls the error in the Trotter/PF evolution upto time $t_m$. The Trotter error can be reduced by choosing smaller $\Delta t$ \cite{childs2021theory}, leading to a smaller global error. Second, $\Delta t$ enters the global error through DMD for $t>t_m$ (Theorem ~\ref{theorem:local_error_quantum}). 
Thus for DMD prediction, the choice of $\Delta t$ is dependent on the initial Trotter steps. In our analysis, we choose the same $\Delta t$ for the initial Trotter evolution up to time $t_m$ and for DMD prediction. However, DMD predictions can be made better if the consecutive column entries in $\tX_1$ and  $\tX_2$ are linearly independent. One can choose the columns for  $\tX_1$ and  $\tX_2$ from the measured dataset ($\mathcal{D}$) by skipping some of the columns so that the time difference between each columns in $\tX_1$ and $\tX_2$ are multiples of $\Delta t$, i.e, $\tau\Delta t$, $\tau=2, 3, ...$. Such measures can reduce $\|\mathbf{\Phi}_m\|_2\|\mathbf{\Phi}_m^{-1}\|_2$ and $\|\epsilon^m\|_2$ in~(\ref{eq:global_error_bound}) and can give better DMD approximations. 
Taking into account these impacts of $\Delta t$, in general, a reduction in $\Delta t$ results in a decrease in the global error. However, this improvement in accuracy comes at the cost of growing computational cost on the quantum device. The selection of an appropriate $\Delta t$ involves a trade-off between accuracy and efficiency.
It should be noted that DMD may be combined with other quantum dynamics simulation methods as well (e.g, qubitization, linear combination of unitaries, adaptive and variational methods).  It predicts observables in the long run based on limited time dynamics simulation. Total systematic error will depend on the choice of the device simulation method.

\bibliographystyle{apsrev4-2}
\bibliography{sample}

\begin{thebibliography}{64}%
\makeatletter
\providecommand \@ifxundefined [1]{%
 \@ifx{#1\undefined}
}%
\providecommand \@ifnum [1]{%
 \ifnum #1\expandafter \@firstoftwo
 \else \expandafter \@secondoftwo
 \fi
}%
\providecommand \@ifx [1]{%
 \ifx #1\expandafter \@firstoftwo
 \else \expandafter \@secondoftwo
 \fi
}%
\providecommand \natexlab [1]{#1}%
\providecommand \enquote  [1]{``#1''}%
\providecommand \bibnamefont  [1]{#1}%
\providecommand \bibfnamefont [1]{#1}%
\providecommand \citenamefont [1]{#1}%
\providecommand \href@noop [0]{\@secondoftwo}%
\providecommand \href [0]{\begingroup \@sanitize@url \@href}%
\providecommand \@href[1]{\@@startlink{#1}\@@href}%
\providecommand \@@href[1]{\endgroup#1\@@endlink}%
\providecommand \@sanitize@url [0]{\catcode `\\12\catcode `\$12\catcode
  `\&12\catcode `\#12\catcode `\^12\catcode `\_12\catcode `\%12\relax}%
\providecommand \@@startlink[1]{}%
\providecommand \@@endlink[0]{}%
\providecommand \url  [0]{\begingroup\@sanitize@url \@url }%
\providecommand \@url [1]{\endgroup\@href {#1}{\urlprefix }}%
\providecommand \urlprefix  [0]{URL }%
\providecommand \Eprint [0]{\href }%
\providecommand \doibase [0]{https://doi.org/}%
\providecommand \selectlanguage [0]{\@gobble}%
\providecommand \bibinfo  [0]{\@secondoftwo}%
\providecommand \bibfield  [0]{\@secondoftwo}%
\providecommand \translation [1]{[#1]}%
\providecommand \BibitemOpen [0]{}%
\providecommand \bibitemStop [0]{}%
\providecommand \bibitemNoStop [0]{.\EOS\space}%
\providecommand \EOS [0]{\spacefactor3000\relax}%
\providecommand \BibitemShut  [1]{\csname bibitem#1\endcsname}%
\let\auto@bib@innerbib\@empty
\bibitem [{\citenamefont {Smith}\ \emph {et~al.}(2019)\citenamefont {Smith},
  \citenamefont {Kim}, \citenamefont {Pollmann},\ and\ \citenamefont
  {Knolle}}]{Smith_2019}%
  \BibitemOpen
  \bibfield  {author} {\bibinfo {author} {\bibfnamefont {A.}~\bibnamefont
  {Smith}}, \bibinfo {author} {\bibfnamefont {M.~S.}\ \bibnamefont {Kim}},
  \bibinfo {author} {\bibfnamefont {F.}~\bibnamefont {Pollmann}},\ and\
  \bibinfo {author} {\bibfnamefont {J.}~\bibnamefont {Knolle}},\ }\bibfield
  {journal} {\bibinfo  {journal} {npj Quantum Information}\ }\textbf {\bibinfo
  {volume} {5}},\ \href {https://doi.org/10.1038/s41534-019-0217-0}
  {10.1038/s41534-019-0217-0} (\bibinfo {year} {2019})\BibitemShut {NoStop}%
\bibitem [{\citenamefont {Yao}\ \emph {et~al.}(2021)\citenamefont {Yao},
  \citenamefont {Gomes}, \citenamefont {Zhang}, \citenamefont {Wang},
  \citenamefont {Ho}, \citenamefont {Iadecola},\ and\ \citenamefont
  {Orth}}]{AVQDS}%
  \BibitemOpen
  \bibfield  {author} {\bibinfo {author} {\bibfnamefont {Y.-X.}\ \bibnamefont
  {Yao}}, \bibinfo {author} {\bibfnamefont {N.}~\bibnamefont {Gomes}}, \bibinfo
  {author} {\bibfnamefont {F.}~\bibnamefont {Zhang}}, \bibinfo {author}
  {\bibfnamefont {C.-Z.}\ \bibnamefont {Wang}}, \bibinfo {author}
  {\bibfnamefont {K.-M.}\ \bibnamefont {Ho}}, \bibinfo {author} {\bibfnamefont
  {T.}~\bibnamefont {Iadecola}},\ and\ \bibinfo {author} {\bibfnamefont
  {P.~P.}\ \bibnamefont {Orth}},\ }\href
  {https://doi.org/10.1103/PRXQuantum.2.030307} {\bibfield  {journal} {\bibinfo
   {journal} {PRX Quantum}\ }\textbf {\bibinfo {volume} {2}},\ \bibinfo {pages}
  {030307} (\bibinfo {year} {2021})}\BibitemShut {NoStop}%
\bibitem [{\citenamefont {Lee}\ \emph {et~al.}(2022)\citenamefont {Lee},
  \citenamefont {Hsieh}, \citenamefont {Zhang},\ and\ \citenamefont
  {Shi}}]{lee2022variational}%
  \BibitemOpen
  \bibfield  {author} {\bibinfo {author} {\bibfnamefont {C.-K.}\ \bibnamefont
  {Lee}}, \bibinfo {author} {\bibfnamefont {C.-Y.}\ \bibnamefont {Hsieh}},
  \bibinfo {author} {\bibfnamefont {S.}~\bibnamefont {Zhang}},\ and\ \bibinfo
  {author} {\bibfnamefont {L.}~\bibnamefont {Shi}},\ }\href@noop {} {\bibfield
  {journal} {\bibinfo  {journal} {Journal of Chemical Theory and Computation}\
  }\textbf {\bibinfo {volume} {18}},\ \bibinfo {pages} {2105} (\bibinfo {year}
  {2022})}\BibitemShut {NoStop}%
\bibitem [{\citenamefont {Preskill}(2018)}]{nisq}%
  \BibitemOpen
  \bibfield  {author} {\bibinfo {author} {\bibfnamefont {J.}~\bibnamefont
  {Preskill}},\ }\href@noop {} {\bibfield  {journal} {\bibinfo  {journal}
  {Quantum}\ }\textbf {\bibinfo {volume} {2}},\ \bibinfo {pages} {79} (\bibinfo
  {year} {2018})}\BibitemShut {NoStop}%
\bibitem [{\citenamefont {Lee}\ \emph {et~al.}(2023)\citenamefont {Lee},
  \citenamefont {Lee}, \citenamefont {Zhai}, \citenamefont {Tong},
  \citenamefont {Dalzell}, \citenamefont {Kumar}, \citenamefont {Helms},
  \citenamefont {Gray}, \citenamefont {Cui}, \citenamefont {Liu} \emph
  {et~al.}}]{lee2208there}%
  \BibitemOpen
  \bibfield  {author} {\bibinfo {author} {\bibfnamefont {S.}~\bibnamefont
  {Lee}}, \bibinfo {author} {\bibfnamefont {J.}~\bibnamefont {Lee}}, \bibinfo
  {author} {\bibfnamefont {H.}~\bibnamefont {Zhai}}, \bibinfo {author}
  {\bibfnamefont {Y.}~\bibnamefont {Tong}}, \bibinfo {author} {\bibfnamefont
  {A.~M.}\ \bibnamefont {Dalzell}}, \bibinfo {author} {\bibfnamefont
  {A.}~\bibnamefont {Kumar}}, \bibinfo {author} {\bibfnamefont
  {P.}~\bibnamefont {Helms}}, \bibinfo {author} {\bibfnamefont
  {J.}~\bibnamefont {Gray}}, \bibinfo {author} {\bibfnamefont {Z.-H.}\
  \bibnamefont {Cui}}, \bibinfo {author} {\bibfnamefont {W.}~\bibnamefont
  {Liu}}, \emph {et~al.},\ }\href {https://doi.org/10.1038/s41467-023-37587-6}
  {\bibfield  {journal} {\bibinfo  {journal} {Nature Communications}\ }\textbf
  {\bibinfo {volume} {14}},\ \bibinfo {pages} {1952} (\bibinfo {year}
  {2023})}\BibitemShut {NoStop}%
\bibitem [{\citenamefont {Zimbor{\'a}s}\ \emph {et~al.}(2025)\citenamefont
  {Zimbor{\'a}s}, \citenamefont {Koczor}, \citenamefont {Holmes}, \citenamefont
  {Borrelli}, \citenamefont {Gily{\'e}n}, \citenamefont {Huang}, \citenamefont
  {Cai}, \citenamefont {Ac{\'i}n}, \citenamefont {Aolita}, \citenamefont
  {Banchi}, \citenamefont {Brand{\~a}o}, \citenamefont {Cavalcanti},
  \citenamefont {Cubitt}, \citenamefont {Filippov}, \citenamefont
  {Garc{\'i}a-P{\'e}rez}, \citenamefont {Goold}, \citenamefont
  {K{\'a}lm{\'a}n}, \citenamefont {Kyoseva}, \citenamefont {Rossi},
  \citenamefont {Sokolov}, \citenamefont {Tavernelli},\ and\ \citenamefont
  {Maniscalco}}]{zimboras2025myths}%
  \BibitemOpen
  \bibfield  {author} {\bibinfo {author} {\bibfnamefont {Z.}~\bibnamefont
  {Zimbor{\'a}s}}, \bibinfo {author} {\bibfnamefont {B.}~\bibnamefont
  {Koczor}}, \bibinfo {author} {\bibfnamefont {Z.}~\bibnamefont {Holmes}},
  \bibinfo {author} {\bibfnamefont {E.-M.}\ \bibnamefont {Borrelli}}, \bibinfo
  {author} {\bibfnamefont {A.}~\bibnamefont {Gily{\'e}n}}, \bibinfo {author}
  {\bibfnamefont {H.-Y.}\ \bibnamefont {Huang}}, \bibinfo {author}
  {\bibfnamefont {Z.}~\bibnamefont {Cai}}, \bibinfo {author} {\bibfnamefont
  {A.}~\bibnamefont {Ac{\'i}n}}, \bibinfo {author} {\bibfnamefont
  {L.}~\bibnamefont {Aolita}}, \bibinfo {author} {\bibfnamefont
  {L.}~\bibnamefont {Banchi}}, \bibinfo {author} {\bibfnamefont {F.~G. S.~L.}\
  \bibnamefont {Brand{\~a}o}}, \bibinfo {author} {\bibfnamefont
  {D.}~\bibnamefont {Cavalcanti}}, \bibinfo {author} {\bibfnamefont
  {T.}~\bibnamefont {Cubitt}}, \bibinfo {author} {\bibfnamefont {S.~N.}\
  \bibnamefont {Filippov}}, \bibinfo {author} {\bibfnamefont {G.}~\bibnamefont
  {Garc{\'i}a-P{\'e}rez}}, \bibinfo {author} {\bibfnamefont {J.}~\bibnamefont
  {Goold}}, \bibinfo {author} {\bibfnamefont {O.}~\bibnamefont
  {K{\'a}lm{\'a}n}}, \bibinfo {author} {\bibfnamefont {E.}~\bibnamefont
  {Kyoseva}}, \bibinfo {author} {\bibfnamefont {M.~A.~C.}\ \bibnamefont
  {Rossi}}, \bibinfo {author} {\bibfnamefont {B.}~\bibnamefont {Sokolov}},
  \bibinfo {author} {\bibfnamefont {I.}~\bibnamefont {Tavernelli}},\ and\
  \bibinfo {author} {\bibfnamefont {S.}~\bibnamefont {Maniscalco}},\
  }\href@noop {} {\bibfield  {journal} {\bibinfo  {journal} {arXiv preprint
  arXiv:2501.05694}\ } (\bibinfo {year} {2025})},\ \Eprint
  {https://arxiv.org/abs/2501.05694} {arXiv:2501.05694 [quant-ph]} \BibitemShut
  {NoStop}%
\bibitem [{\citenamefont {McLachlan}(1964)}]{variational_mclachlan}%
  \BibitemOpen
  \bibfield  {author} {\bibinfo {author} {\bibfnamefont {A.}~\bibnamefont
  {McLachlan}},\ }\href {https://doi.org/10.1080/00268976400100041} {\bibfield
  {journal} {\bibinfo  {journal} {Molecular Physics}\ }\textbf {\bibinfo
  {volume} {8}},\ \bibinfo {pages} {39} (\bibinfo {year} {1964})},\ \Eprint
  {https://arxiv.org/abs/https://doi.org/10.1080/00268976400100041}
  {https://doi.org/10.1080/00268976400100041} \BibitemShut {NoStop}%
\bibitem [{\citenamefont {Peruzzo}\ \emph {et~al.}(2014)\citenamefont
  {Peruzzo}, \citenamefont {McClean}, \citenamefont {Shadbolt}, \citenamefont
  {Yung}, \citenamefont {Zhou}, \citenamefont {Love}, \citenamefont
  {Aspuru-Guzik},\ and\ \citenamefont {O'Brien}}]{vqe}%
  \BibitemOpen
  \bibfield  {author} {\bibinfo {author} {\bibfnamefont {A.}~\bibnamefont
  {Peruzzo}}, \bibinfo {author} {\bibfnamefont {J.}~\bibnamefont {McClean}},
  \bibinfo {author} {\bibfnamefont {P.}~\bibnamefont {Shadbolt}}, \bibinfo
  {author} {\bibfnamefont {M.-H.}\ \bibnamefont {Yung}}, \bibinfo {author}
  {\bibfnamefont {X.-Q.}\ \bibnamefont {Zhou}}, \bibinfo {author}
  {\bibfnamefont {P.~J.}\ \bibnamefont {Love}}, \bibinfo {author}
  {\bibfnamefont {A.}~\bibnamefont {Aspuru-Guzik}},\ and\ \bibinfo {author}
  {\bibfnamefont {J.~L.}\ \bibnamefont {O'Brien}},\ }\href@noop {} {\bibfield
  {journal} {\bibinfo  {journal} {Nuovo Cimento}\ }\textbf {\bibinfo {volume}
  {5}},\ \bibinfo {pages} {4213} (\bibinfo {year} {2014})}\BibitemShut
  {NoStop}%
\bibitem [{\citenamefont {McClean}\ \emph {et~al.}(2016)\citenamefont
  {McClean}, \citenamefont {Romero}, \citenamefont {Babbush},\ and\
  \citenamefont {Aspuru-Guzik}}]{vqe_theory}%
  \BibitemOpen
  \bibfield  {author} {\bibinfo {author} {\bibfnamefont {J.~R.}\ \bibnamefont
  {McClean}}, \bibinfo {author} {\bibfnamefont {J.}~\bibnamefont {Romero}},
  \bibinfo {author} {\bibfnamefont {R.}~\bibnamefont {Babbush}},\ and\ \bibinfo
  {author} {\bibfnamefont {A.}~\bibnamefont {Aspuru-Guzik}},\ }\href
  {https://doi.org/10.1088/1367-2630/18/2/023023} {\bibfield  {journal}
  {\bibinfo  {journal} {New Journal of Physics}\ }\textbf {\bibinfo {volume}
  {18}},\ \bibinfo {pages} {023023} (\bibinfo {year} {2016})}\BibitemShut
  {NoStop}%
\bibitem [{\citenamefont {O'Malley}\ \emph {et~al.}(2016)\citenamefont
  {O'Malley}, \citenamefont {Babbush}, \citenamefont {Kivlichan}, \citenamefont
  {Romero}, \citenamefont {McClean}, \citenamefont {Barends}, \citenamefont
  {Kelly}, \citenamefont {Roushan}, \citenamefont {Tranter}, \citenamefont
  {Ding}, \citenamefont {Campbell}, \citenamefont {Chen}, \citenamefont {Chen},
  \citenamefont {Chiaro}, \citenamefont {Dunsworth}, \citenamefont {Fowler},
  \citenamefont {Jeffrey}, \citenamefont {Lucero}, \citenamefont {Megrant},
  \citenamefont {Mutus}, \citenamefont {Neeley}, \citenamefont {Neill},
  \citenamefont {Quintana}, \citenamefont {Sank}, \citenamefont {Vainsencher},
  \citenamefont {Wenner}, \citenamefont {White}, \citenamefont {Coveney},
  \citenamefont {Love}, \citenamefont {Neven}, \citenamefont {Aspuru-Guzik},\
  and\ \citenamefont {Martinis}}]{vqe_pea_h2}%
  \BibitemOpen
  \bibfield  {author} {\bibinfo {author} {\bibfnamefont {P.~J.~J.}\
  \bibnamefont {O'Malley}}, \bibinfo {author} {\bibfnamefont {R.}~\bibnamefont
  {Babbush}}, \bibinfo {author} {\bibfnamefont {I.~D.}\ \bibnamefont
  {Kivlichan}}, \bibinfo {author} {\bibfnamefont {J.}~\bibnamefont {Romero}},
  \bibinfo {author} {\bibfnamefont {J.~R.}\ \bibnamefont {McClean}}, \bibinfo
  {author} {\bibfnamefont {R.}~\bibnamefont {Barends}}, \bibinfo {author}
  {\bibfnamefont {J.}~\bibnamefont {Kelly}}, \bibinfo {author} {\bibfnamefont
  {P.}~\bibnamefont {Roushan}}, \bibinfo {author} {\bibfnamefont
  {A.}~\bibnamefont {Tranter}}, \bibinfo {author} {\bibfnamefont
  {N.}~\bibnamefont {Ding}}, \bibinfo {author} {\bibfnamefont {B.}~\bibnamefont
  {Campbell}}, \bibinfo {author} {\bibfnamefont {Y.}~\bibnamefont {Chen}},
  \bibinfo {author} {\bibfnamefont {Z.}~\bibnamefont {Chen}}, \bibinfo {author}
  {\bibfnamefont {B.}~\bibnamefont {Chiaro}}, \bibinfo {author} {\bibfnamefont
  {A.}~\bibnamefont {Dunsworth}}, \bibinfo {author} {\bibfnamefont {A.~G.}\
  \bibnamefont {Fowler}}, \bibinfo {author} {\bibfnamefont {E.}~\bibnamefont
  {Jeffrey}}, \bibinfo {author} {\bibfnamefont {E.}~\bibnamefont {Lucero}},
  \bibinfo {author} {\bibfnamefont {A.}~\bibnamefont {Megrant}}, \bibinfo
  {author} {\bibfnamefont {J.~Y.}\ \bibnamefont {Mutus}}, \bibinfo {author}
  {\bibfnamefont {M.}~\bibnamefont {Neeley}}, \bibinfo {author} {\bibfnamefont
  {C.}~\bibnamefont {Neill}}, \bibinfo {author} {\bibfnamefont
  {C.}~\bibnamefont {Quintana}}, \bibinfo {author} {\bibfnamefont
  {D.}~\bibnamefont {Sank}}, \bibinfo {author} {\bibfnamefont {A.}~\bibnamefont
  {Vainsencher}}, \bibinfo {author} {\bibfnamefont {J.}~\bibnamefont {Wenner}},
  \bibinfo {author} {\bibfnamefont {T.~C.}\ \bibnamefont {White}}, \bibinfo
  {author} {\bibfnamefont {P.~V.}\ \bibnamefont {Coveney}}, \bibinfo {author}
  {\bibfnamefont {P.~J.}\ \bibnamefont {Love}}, \bibinfo {author}
  {\bibfnamefont {H.}~\bibnamefont {Neven}}, \bibinfo {author} {\bibfnamefont
  {A.}~\bibnamefont {Aspuru-Guzik}},\ and\ \bibinfo {author} {\bibfnamefont
  {J.~M.}\ \bibnamefont {Martinis}},\ }\href
  {https://doi.org/10.1103/PhysRevX.6.031007} {\bibfield  {journal} {\bibinfo
  {journal} {Phys. Rev. X}\ }\textbf {\bibinfo {volume} {6}},\ \bibinfo {pages}
  {031007} (\bibinfo {year} {2016})}\BibitemShut {NoStop}%
\bibitem [{\citenamefont {Kandala}\ \emph {et~al.}(2017)\citenamefont
  {Kandala}, \citenamefont {Mezzacapo}, \citenamefont {Temme}, \citenamefont
  {Takita}, \citenamefont {Brink}, \citenamefont {Chow},\ and\ \citenamefont
  {Gambetta}}]{hardware_efficient_vqe}%
  \BibitemOpen
  \bibfield  {author} {\bibinfo {author} {\bibfnamefont {A.}~\bibnamefont
  {Kandala}}, \bibinfo {author} {\bibfnamefont {A.}~\bibnamefont {Mezzacapo}},
  \bibinfo {author} {\bibfnamefont {K.}~\bibnamefont {Temme}}, \bibinfo
  {author} {\bibfnamefont {M.}~\bibnamefont {Takita}}, \bibinfo {author}
  {\bibfnamefont {M.}~\bibnamefont {Brink}}, \bibinfo {author} {\bibfnamefont
  {J.~M.}\ \bibnamefont {Chow}},\ and\ \bibinfo {author} {\bibfnamefont
  {J.~M.}\ \bibnamefont {Gambetta}},\ }\href@noop {} {\bibfield  {journal}
  {\bibinfo  {journal} {Nature}\ }\textbf {\bibinfo {volume} {549}},\ \bibinfo
  {pages} {242} (\bibinfo {year} {2017})}\BibitemShut {NoStop}%
\bibitem [{\citenamefont {Ryabinkin}\ \emph {et~al.}(2018)\citenamefont
  {Ryabinkin}, \citenamefont {Yen}, \citenamefont {Genin},\ and\ \citenamefont
  {Izmaylov}}]{VQE_qcc}%
  \BibitemOpen
  \bibfield  {author} {\bibinfo {author} {\bibfnamefont {I.~G.}\ \bibnamefont
  {Ryabinkin}}, \bibinfo {author} {\bibfnamefont {T.-C.}\ \bibnamefont {Yen}},
  \bibinfo {author} {\bibfnamefont {S.~N.}\ \bibnamefont {Genin}},\ and\
  \bibinfo {author} {\bibfnamefont {A.~F.}\ \bibnamefont {Izmaylov}},\
  }\href@noop {} {\bibfield  {journal} {\bibinfo  {journal} {J. Chem. Theory
  Comput.}\ }\textbf {\bibinfo {volume} {14}},\ \bibinfo {pages} {6317}
  (\bibinfo {year} {2018})}\BibitemShut {NoStop}%
\bibitem [{\citenamefont {Zhang}\ \emph {et~al.}(2021)\citenamefont {Zhang},
  \citenamefont {Gomes}, \citenamefont {Berthusen}, \citenamefont {Orth},
  \citenamefont {Wang}, \citenamefont {Ho},\ and\ \citenamefont
  {Yao}}]{FengVQE}%
  \BibitemOpen
  \bibfield  {author} {\bibinfo {author} {\bibfnamefont {F.}~\bibnamefont
  {Zhang}}, \bibinfo {author} {\bibfnamefont {N.}~\bibnamefont {Gomes}},
  \bibinfo {author} {\bibfnamefont {N.~F.}\ \bibnamefont {Berthusen}}, \bibinfo
  {author} {\bibfnamefont {P.~P.}\ \bibnamefont {Orth}}, \bibinfo {author}
  {\bibfnamefont {C.-Z.}\ \bibnamefont {Wang}}, \bibinfo {author}
  {\bibfnamefont {K.-M.}\ \bibnamefont {Ho}},\ and\ \bibinfo {author}
  {\bibfnamefont {Y.-X.}\ \bibnamefont {Yao}},\ }\href@noop {} {\bibfield
  {journal} {\bibinfo  {journal} {Physical Review Research}\ }\textbf {\bibinfo
  {volume} {3}},\ \bibinfo {pages} {013039} (\bibinfo {year}
  {2021})}\BibitemShut {NoStop}%
\bibitem [{\citenamefont {Huang}\ \emph {et~al.}(2020)\citenamefont {Huang},
  \citenamefont {Kueng},\ and\ \citenamefont {Preskill}}]{huang2020predicting}%
  \BibitemOpen
  \bibfield  {author} {\bibinfo {author} {\bibfnamefont {H.-Y.}\ \bibnamefont
  {Huang}}, \bibinfo {author} {\bibfnamefont {R.}~\bibnamefont {Kueng}},\ and\
  \bibinfo {author} {\bibfnamefont {J.}~\bibnamefont {Preskill}},\ }\href@noop
  {} {\bibfield  {journal} {\bibinfo  {journal} {Nature Physics}\ }\textbf
  {\bibinfo {volume} {16}},\ \bibinfo {pages} {1050} (\bibinfo {year}
  {2020})}\BibitemShut {NoStop}%
\bibitem [{\citenamefont {Huang}\ \emph {et~al.}(2021)\citenamefont {Huang},
  \citenamefont {Kueng},\ and\ \citenamefont {Preskill}}]{huang2021efficient}%
  \BibitemOpen
  \bibfield  {author} {\bibinfo {author} {\bibfnamefont {H.-Y.}\ \bibnamefont
  {Huang}}, \bibinfo {author} {\bibfnamefont {R.}~\bibnamefont {Kueng}},\ and\
  \bibinfo {author} {\bibfnamefont {J.}~\bibnamefont {Preskill}},\ }\href@noop
  {} {\bibfield  {journal} {\bibinfo  {journal} {Physical review letters}\
  }\textbf {\bibinfo {volume} {127}},\ \bibinfo {pages} {030503} (\bibinfo
  {year} {2021})}\BibitemShut {NoStop}%
\bibitem [{\citenamefont {Lloyd}(1996)}]{lloyd1996universal}%
  \BibitemOpen
  \bibfield  {author} {\bibinfo {author} {\bibfnamefont {S.}~\bibnamefont
  {Lloyd}},\ }\href@noop {} {\bibfield  {journal} {\bibinfo  {journal}
  {Science}\ }\textbf {\bibinfo {volume} {273}},\ \bibinfo {pages} {1073}
  (\bibinfo {year} {1996})}\BibitemShut {NoStop}%
\bibitem [{\citenamefont {Aharonov}\ and\ \citenamefont
  {Ta-Shma}(2003)}]{aharonov2003adiabatic}%
  \BibitemOpen
  \bibfield  {author} {\bibinfo {author} {\bibfnamefont {D.}~\bibnamefont
  {Aharonov}}\ and\ \bibinfo {author} {\bibfnamefont {A.}~\bibnamefont
  {Ta-Shma}},\ }in\ \href@noop {} {\emph {\bibinfo {booktitle} {Proceedings of
  the thirty-fifth annual ACM symposium on Theory of computing}}}\ (\bibinfo
  {year} {2003})\ pp.\ \bibinfo {pages} {20--29}\BibitemShut {NoStop}%
\bibitem [{\citenamefont {Steckmann}\ \emph {et~al.}(2023)\citenamefont
  {Steckmann}, \citenamefont {Keen}, \citenamefont {Kemper}, \citenamefont
  {Dumitrescu},\ and\ \citenamefont {Wang}}]{Steckmann2021}%
  \BibitemOpen
  \bibfield  {author} {\bibinfo {author} {\bibfnamefont {T.}~\bibnamefont
  {Steckmann}}, \bibinfo {author} {\bibfnamefont {T.}~\bibnamefont {Keen}},
  \bibinfo {author} {\bibfnamefont {A.~F.}\ \bibnamefont {Kemper}}, \bibinfo
  {author} {\bibfnamefont {E.~F.}\ \bibnamefont {Dumitrescu}},\ and\ \bibinfo
  {author} {\bibfnamefont {Y.}~\bibnamefont {Wang}},\ }\href
  {https://doi.org/10.1103/PhysRevResearch.5.023198} {\bibfield  {journal}
  {\bibinfo  {journal} {Physical Review Research}\ }\textbf {\bibinfo {volume}
  {5}},\ \bibinfo {pages} {023198} (\bibinfo {year} {2023})}\BibitemShut
  {NoStop}%
\bibitem [{\citenamefont {Berry}\ \emph {et~al.}(2007)\citenamefont {Berry},
  \citenamefont {Ahokas}, \citenamefont {Cleve},\ and\ \citenamefont
  {Sanders}}]{berry2007efficient}%
  \BibitemOpen
  \bibfield  {author} {\bibinfo {author} {\bibfnamefont {D.~W.}\ \bibnamefont
  {Berry}}, \bibinfo {author} {\bibfnamefont {G.}~\bibnamefont {Ahokas}},
  \bibinfo {author} {\bibfnamefont {R.}~\bibnamefont {Cleve}},\ and\ \bibinfo
  {author} {\bibfnamefont {B.~C.}\ \bibnamefont {Sanders}},\ }\href@noop {}
  {\bibfield  {journal} {\bibinfo  {journal} {Communications in Mathematical
  Physics}\ }\textbf {\bibinfo {volume} {270}},\ \bibinfo {pages} {359}
  (\bibinfo {year} {2007})}\BibitemShut {NoStop}%
\bibitem [{\citenamefont {Low}\ and\ \citenamefont
  {Chuang}(2017)}]{low2017optimal}%
  \BibitemOpen
  \bibfield  {author} {\bibinfo {author} {\bibfnamefont {G.~H.}\ \bibnamefont
  {Low}}\ and\ \bibinfo {author} {\bibfnamefont {I.~L.}\ \bibnamefont
  {Chuang}},\ }\href {https://doi.org/10.1103/PhysRevLett.118.010501}
  {\bibfield  {journal} {\bibinfo  {journal} {Physical Review Letters}\
  }\textbf {\bibinfo {volume} {118}},\ \bibinfo {pages} {010501} (\bibinfo
  {year} {2017})}\BibitemShut {NoStop}%
\bibitem [{\citenamefont {Low}\ and\ \citenamefont
  {Chuang}(2019)}]{low2019hamiltonian}%
  \BibitemOpen
  \bibfield  {author} {\bibinfo {author} {\bibfnamefont {G.~H.}\ \bibnamefont
  {Low}}\ and\ \bibinfo {author} {\bibfnamefont {I.~L.}\ \bibnamefont
  {Chuang}},\ }\href {https://doi.org/10.22331/q-2019-07-12-163} {\bibfield
  {journal} {\bibinfo  {journal} {Quantum}\ }\textbf {\bibinfo {volume} {3}},\
  \bibinfo {pages} {163} (\bibinfo {year} {2019})}\BibitemShut {NoStop}%
\bibitem [{\citenamefont {Mohseni}\ \emph {et~al.}(2022)\citenamefont
  {Mohseni}, \citenamefont {F{\"o}sel}, \citenamefont {Guo}, \citenamefont
  {Navarrete-Benlloch},\ and\ \citenamefont {Marquardt}}]{mohseni2022deep}%
  \BibitemOpen
  \bibfield  {author} {\bibinfo {author} {\bibfnamefont {N.}~\bibnamefont
  {Mohseni}}, \bibinfo {author} {\bibfnamefont {T.}~\bibnamefont {F{\"o}sel}},
  \bibinfo {author} {\bibfnamefont {L.}~\bibnamefont {Guo}}, \bibinfo {author}
  {\bibfnamefont {C.}~\bibnamefont {Navarrete-Benlloch}},\ and\ \bibinfo
  {author} {\bibfnamefont {F.}~\bibnamefont {Marquardt}},\ }\href@noop {}
  {\bibfield  {journal} {\bibinfo  {journal} {Quantum}\ }\textbf {\bibinfo
  {volume} {6}},\ \bibinfo {pages} {714} (\bibinfo {year} {2022})}\BibitemShut
  {NoStop}%
\bibitem [{\citenamefont {Mohseni}\ \emph {et~al.}(2024)\citenamefont
  {Mohseni}, \citenamefont {Shi}, \citenamefont {Byrnes},\ and\ \citenamefont
  {Hartmann}}]{mohseni2024deep}%
  \BibitemOpen
  \bibfield  {author} {\bibinfo {author} {\bibfnamefont {N.}~\bibnamefont
  {Mohseni}}, \bibinfo {author} {\bibfnamefont {J.}~\bibnamefont {Shi}},
  \bibinfo {author} {\bibfnamefont {T.}~\bibnamefont {Byrnes}},\ and\ \bibinfo
  {author} {\bibfnamefont {M.~J.}\ \bibnamefont {Hartmann}},\ }\href@noop {}
  {\bibfield  {journal} {\bibinfo  {journal} {Quantum}\ }\textbf {\bibinfo
  {volume} {8}},\ \bibinfo {pages} {1417} (\bibinfo {year} {2024})}\BibitemShut
  {NoStop}%
\bibitem [{\citenamefont {Cirstoiu}\ \emph {et~al.}(2020)\citenamefont
  {Cirstoiu}, \citenamefont {Holmes}, \citenamefont {Iosue}, \citenamefont
  {Cincio}, \citenamefont {Coles},\ and\ \citenamefont
  {Sornborger}}]{cirstoiu2020variational}%
  \BibitemOpen
  \bibfield  {author} {\bibinfo {author} {\bibfnamefont {C.}~\bibnamefont
  {Cirstoiu}}, \bibinfo {author} {\bibfnamefont {Z.}~\bibnamefont {Holmes}},
  \bibinfo {author} {\bibfnamefont {J.}~\bibnamefont {Iosue}}, \bibinfo
  {author} {\bibfnamefont {L.}~\bibnamefont {Cincio}}, \bibinfo {author}
  {\bibfnamefont {P.~J.}\ \bibnamefont {Coles}},\ and\ \bibinfo {author}
  {\bibfnamefont {A.}~\bibnamefont {Sornborger}},\ }\href
  {https://doi.org/10.1038/s41534-020-00302-0} {\bibfield  {journal} {\bibinfo
  {journal} {npj Quantum Information}\ }\textbf {\bibinfo {volume} {6}},\
  \bibinfo {pages} {1} (\bibinfo {year} {2020})}\BibitemShut {NoStop}%
\bibitem [{\citenamefont {Cortes}\ \emph {et~al.}(2022)\citenamefont {Cortes},
  \citenamefont {DePrince},\ and\ \citenamefont
  {Gray}}]{cortes2022fastforwarding}%
  \BibitemOpen
  \bibfield  {author} {\bibinfo {author} {\bibfnamefont {C.~L.}\ \bibnamefont
  {Cortes}}, \bibinfo {author} {\bibfnamefont {A.~E.}\ \bibnamefont
  {DePrince}},\ and\ \bibinfo {author} {\bibfnamefont {S.~K.}\ \bibnamefont
  {Gray}},\ }\href@noop {} {\bibfield  {journal} {\bibinfo  {journal} {arXiv
  preprint arXiv:2208.00948}\ } (\bibinfo {year} {2022})},\ \Eprint
  {https://arxiv.org/abs/2208.00948} {arXiv:2208.00948 [quant-ph]} \BibitemShut
  {NoStop}%
\bibitem [{\citenamefont {Arbabi}\ and\ \citenamefont
  {Mezi\'{c}}(2017)}]{DMDtoKoop}%
  \BibitemOpen
  \bibfield  {author} {\bibinfo {author} {\bibfnamefont {H.}~\bibnamefont
  {Arbabi}}\ and\ \bibinfo {author} {\bibfnamefont {I.}~\bibnamefont
  {Mezi\'{c}}},\ }\href {https://epubs.siam.org/doi/10.1137/17M1125236}
  {\bibfield  {journal} {\bibinfo  {journal} {SIAM J. Appl. Dyn. Syst.}\
  }\textbf {\bibinfo {volume} {16}},\ \bibinfo {pages} {2096} (\bibinfo {year}
  {2017})}\BibitemShut {NoStop}%
\bibitem [{\citenamefont {Koopman}(1931)}]{Koopman1}%
  \BibitemOpen
  \bibfield  {author} {\bibinfo {author} {\bibfnamefont {B.~O.}\ \bibnamefont
  {Koopman}},\ }\href {https://www.ncbi.nlm.nih.gov/pmc/articles/PMC1076052/}
  {\bibfield  {journal} {\bibinfo  {journal} {Proc. Natl. Acad. Sci. U.S.A.}\
  }\textbf {\bibinfo {volume} {17}},\ \bibinfo {pages} {315} (\bibinfo {year}
  {1931})}\BibitemShut {NoStop}%
\bibitem [{\citenamefont {Koopman}\ and\ \citenamefont
  {Neumann}(1932)}]{Koopman2}%
  \BibitemOpen
  \bibfield  {author} {\bibinfo {author} {\bibfnamefont {B.~O.}\ \bibnamefont
  {Koopman}}\ and\ \bibinfo {author} {\bibfnamefont {J.~v.}\ \bibnamefont
  {Neumann}},\ }\href {https://www.jstor.org/stable/86259} {\bibfield
  {journal} {\bibinfo  {journal} {Proc. Natl. Acad. Sci. U.S.A.}\ }\textbf
  {\bibinfo {volume} {18}},\ \bibinfo {pages} {255} (\bibinfo {year}
  {1932})}\BibitemShut {NoStop}%
\bibitem [{\citenamefont {Golub}\ and\ \citenamefont {Van~Loan}(2013)}]{SVD}%
  \BibitemOpen
  \bibfield  {author} {\bibinfo {author} {\bibfnamefont {G.~H.}\ \bibnamefont
  {Golub}}\ and\ \bibinfo {author} {\bibfnamefont {C.~F.}\ \bibnamefont
  {Van~Loan}},\ }\href@noop {} {\emph {\bibinfo {title} {Matrix
  computations}}},\ Vol.~\bibinfo {volume} {3}\ (\bibinfo  {publisher} {JHU
  press},\ \bibinfo {year} {2013})\BibitemShut {NoStop}%
\bibitem [{\citenamefont {Wiebe}\ \emph {et~al.}(2011)\citenamefont {Wiebe},
  \citenamefont {Berry}, \citenamefont {H{\o}yer},\ and\ \citenamefont
  {Sanders}}]{wiebe2011simulating}%
  \BibitemOpen
  \bibfield  {author} {\bibinfo {author} {\bibfnamefont {N.}~\bibnamefont
  {Wiebe}}, \bibinfo {author} {\bibfnamefont {D.~W.}\ \bibnamefont {Berry}},
  \bibinfo {author} {\bibfnamefont {P.}~\bibnamefont {H{\o}yer}},\ and\
  \bibinfo {author} {\bibfnamefont {B.~C.}\ \bibnamefont {Sanders}},\
  }\href@noop {} {\bibfield  {journal} {\bibinfo  {journal} {Journal of Physics
  A: Mathematical and Theoretical}\ }\textbf {\bibinfo {volume} {44}},\
  \bibinfo {pages} {445308} (\bibinfo {year} {2011})}\BibitemShut {NoStop}%
\bibitem [{\citenamefont {Pastori}\ \emph {et~al.}(2022)\citenamefont
  {Pastori}, \citenamefont {Olsacher}, \citenamefont {Kokail},\ and\
  \citenamefont {Zoller}}]{Pastori2024}%
  \BibitemOpen
  \bibfield  {author} {\bibinfo {author} {\bibfnamefont {L.}~\bibnamefont
  {Pastori}}, \bibinfo {author} {\bibfnamefont {T.}~\bibnamefont {Olsacher}},
  \bibinfo {author} {\bibfnamefont {C.}~\bibnamefont {Kokail}},\ and\ \bibinfo
  {author} {\bibfnamefont {P.}~\bibnamefont {Zoller}},\ }\href
  {https://doi.org/10.1103/PRXQuantum.3.030324} {\bibfield  {journal} {\bibinfo
   {journal} {PRX Quantum}\ }\textbf {\bibinfo {volume} {3}},\ \bibinfo {pages}
  {030324} (\bibinfo {year} {2022})}\BibitemShut {NoStop}%
\bibitem [{\citenamefont {Babbush}\ \emph {et~al.}(2015)\citenamefont
  {Babbush}, \citenamefont {McClean}, \citenamefont {Wecker}, \citenamefont
  {Aspuru-Guzik},\ and\ \citenamefont {Wiebe}}]{babbush2015chemical}%
  \BibitemOpen
  \bibfield  {author} {\bibinfo {author} {\bibfnamefont {R.}~\bibnamefont
  {Babbush}}, \bibinfo {author} {\bibfnamefont {J.}~\bibnamefont {McClean}},
  \bibinfo {author} {\bibfnamefont {D.}~\bibnamefont {Wecker}}, \bibinfo
  {author} {\bibfnamefont {A.}~\bibnamefont {Aspuru-Guzik}},\ and\ \bibinfo
  {author} {\bibfnamefont {N.}~\bibnamefont {Wiebe}},\ }\href@noop {}
  {\bibfield  {journal} {\bibinfo  {journal} {Physical Review A}\ }\textbf
  {\bibinfo {volume} {91}},\ \bibinfo {pages} {022311} (\bibinfo {year}
  {2015})}\BibitemShut {NoStop}%
\bibitem [{\citenamefont {Schmid}(2010)}]{DMD0}%
  \BibitemOpen
  \bibfield  {author} {\bibinfo {author} {\bibfnamefont {P.~J.}\ \bibnamefont
  {Schmid}},\ }\href@noop {} {\bibfield  {journal} {\bibinfo  {journal} {J.
  Fluid Mech.}\ }\textbf {\bibinfo {volume} {656}},\ \bibinfo {pages} {5}
  (\bibinfo {year} {2010})}\BibitemShut {NoStop}%
\bibitem [{\citenamefont {Schmid}(2011)}]{schmid2011application}%
  \BibitemOpen
  \bibfield  {author} {\bibinfo {author} {\bibfnamefont {P.~J.}\ \bibnamefont
  {Schmid}},\ }\href@noop {} {\bibfield  {journal} {\bibinfo  {journal} {Exp.
  Fluids}\ }\textbf {\bibinfo {volume} {50}},\ \bibinfo {pages} {1123}
  (\bibinfo {year} {2011})}\BibitemShut {NoStop}%
\bibitem [{\citenamefont {Schmid}\ \emph {et~al.}(2011)\citenamefont {Schmid},
  \citenamefont {Li}, \citenamefont {Juniper},\ and\ \citenamefont
  {Pust}}]{schmid2011applications}%
  \BibitemOpen
  \bibfield  {author} {\bibinfo {author} {\bibfnamefont {P.~J.}\ \bibnamefont
  {Schmid}}, \bibinfo {author} {\bibfnamefont {L.}~\bibnamefont {Li}}, \bibinfo
  {author} {\bibfnamefont {M.~P.}\ \bibnamefont {Juniper}},\ and\ \bibinfo
  {author} {\bibfnamefont {O.}~\bibnamefont {Pust}},\ }\href@noop {} {\bibfield
   {journal} {\bibinfo  {journal} {Theor. Comput. Fluid Dyn.}\ }\textbf
  {\bibinfo {volume} {25}},\ \bibinfo {pages} {249} (\bibinfo {year}
  {2011})}\BibitemShut {NoStop}%
\bibitem [{\citenamefont {Schmid}\ \emph {et~al.}(2012)\citenamefont {Schmid},
  \citenamefont {Violato},\ and\ \citenamefont
  {Scarano}}]{schmid2012decomposition}%
  \BibitemOpen
  \bibfield  {author} {\bibinfo {author} {\bibfnamefont {P.~J.}\ \bibnamefont
  {Schmid}}, \bibinfo {author} {\bibfnamefont {D.}~\bibnamefont {Violato}},\
  and\ \bibinfo {author} {\bibfnamefont {F.}~\bibnamefont {Scarano}},\
  }\href@noop {} {\bibfield  {journal} {\bibinfo  {journal} {Exp. Fluids}\
  }\textbf {\bibinfo {volume} {52}},\ \bibinfo {pages} {1567} (\bibinfo {year}
  {2012})}\BibitemShut {NoStop}%
\bibitem [{\citenamefont {Seena}\ and\ \citenamefont
  {Sung}(2011)}]{seena2011dynamic}%
  \BibitemOpen
  \bibfield  {author} {\bibinfo {author} {\bibfnamefont {A.}~\bibnamefont
  {Seena}}\ and\ \bibinfo {author} {\bibfnamefont {H.~J.}\ \bibnamefont
  {Sung}},\ }\href@noop {} {\bibfield  {journal} {\bibinfo  {journal} {Int. J.
  Heat Fluid Flow}\ }\textbf {\bibinfo {volume} {32}},\ \bibinfo {pages} {1098}
  (\bibinfo {year} {2011})}\BibitemShut {NoStop}%
\bibitem [{\citenamefont {Lu}\ and\ \citenamefont
  {Tartakovsky}(2020)}]{lu2020prediction}%
  \BibitemOpen
  \bibfield  {author} {\bibinfo {author} {\bibfnamefont {H.}~\bibnamefont
  {Lu}}\ and\ \bibinfo {author} {\bibfnamefont {D.~M.}\ \bibnamefont
  {Tartakovsky}},\ }\href@noop {} {\bibfield  {journal} {\bibinfo  {journal}
  {SIAM Journal on Scientific Computing}\ }\textbf {\bibinfo {volume} {42}},\
  \bibinfo {pages} {A1639} (\bibinfo {year} {2020})}\BibitemShut {NoStop}%
\bibitem [{\citenamefont {Drmac}\ \emph {et~al.}(2018)\citenamefont {Drmac},
  \citenamefont {Mezic},\ and\ \citenamefont {Mohr}}]{drmac2018data}%
  \BibitemOpen
  \bibfield  {author} {\bibinfo {author} {\bibfnamefont {Z.}~\bibnamefont
  {Drmac}}, \bibinfo {author} {\bibfnamefont {I.}~\bibnamefont {Mezic}},\ and\
  \bibinfo {author} {\bibfnamefont {R.}~\bibnamefont {Mohr}},\ }\href@noop {}
  {\bibfield  {journal} {\bibinfo  {journal} {SIAM Journal on Scientific
  Computing}\ }\textbf {\bibinfo {volume} {40}},\ \bibinfo {pages} {A2253}
  (\bibinfo {year} {2018})}\BibitemShut {NoStop}%
\bibitem [{\citenamefont {B{\"o}ttcher}\ and\ \citenamefont
  {Wenzel}(2008)}]{bottcher2008frobenius}%
  \BibitemOpen
  \bibfield  {author} {\bibinfo {author} {\bibfnamefont {A.}~\bibnamefont
  {B{\"o}ttcher}}\ and\ \bibinfo {author} {\bibfnamefont {D.}~\bibnamefont
  {Wenzel}},\ }\href@noop {} {\bibfield  {journal} {\bibinfo  {journal} {Linear
  algebra and its applications}\ }\textbf {\bibinfo {volume} {429}},\ \bibinfo
  {pages} {1864} (\bibinfo {year} {2008})}\BibitemShut {NoStop}%
\bibitem [{\citenamefont {Wedin}(1973)}]{wedin1973perturbation}%
  \BibitemOpen
  \bibfield  {author} {\bibinfo {author} {\bibfnamefont {P.-{\AA}.}\
  \bibnamefont {Wedin}},\ }\href@noop {} {\bibfield  {journal} {\bibinfo
  {journal} {BIT Numerical Mathematics}\ }\textbf {\bibinfo {volume} {13}},\
  \bibinfo {pages} {217} (\bibinfo {year} {1973})}\BibitemShut {NoStop}%
\bibitem [{\citenamefont {Stewart}(1977)}]{stewart1977perturbation}%
  \BibitemOpen
  \bibfield  {author} {\bibinfo {author} {\bibfnamefont {G.~W.}\ \bibnamefont
  {Stewart}},\ }\href@noop {} {\bibfield  {journal} {\bibinfo  {journal} {SIAM
  review}\ }\textbf {\bibinfo {volume} {19}},\ \bibinfo {pages} {634} (\bibinfo
  {year} {1977})}\BibitemShut {NoStop}%
\bibitem [{\citenamefont {Childs}\ \emph {et~al.}(2021)\citenamefont {Childs},
  \citenamefont {Su}, \citenamefont {Tran}, \citenamefont {Wiebe},\ and\
  \citenamefont {Zhu}}]{childs2021theory}%
  \BibitemOpen
  \bibfield  {author} {\bibinfo {author} {\bibfnamefont {A.~M.}\ \bibnamefont
  {Childs}}, \bibinfo {author} {\bibfnamefont {Y.}~\bibnamefont {Su}}, \bibinfo
  {author} {\bibfnamefont {M.~C.}\ \bibnamefont {Tran}}, \bibinfo {author}
  {\bibfnamefont {N.}~\bibnamefont {Wiebe}},\ and\ \bibinfo {author}
  {\bibfnamefont {S.}~\bibnamefont {Zhu}},\ }\href@noop {} {\bibfield
  {journal} {\bibinfo  {journal} {Physical Review X}\ }\textbf {\bibinfo
  {volume} {11}},\ \bibinfo {pages} {011020} (\bibinfo {year}
  {2021})}\BibitemShut {NoStop}%
\bibitem [{\citenamefont {Cazalilla}\ and\ \citenamefont
  {Rigol}(2010)}]{cazalilla2010focus}%
  \BibitemOpen
  \bibfield  {author} {\bibinfo {author} {\bibfnamefont {M.}~\bibnamefont
  {Cazalilla}}\ and\ \bibinfo {author} {\bibfnamefont {M.}~\bibnamefont
  {Rigol}},\ }\href@noop {} {\bibfield  {journal} {\bibinfo  {journal} {New
  Journal of Physics}\ }\textbf {\bibinfo {volume} {12}},\ \bibinfo {pages}
  {055006} (\bibinfo {year} {2010})}\BibitemShut {NoStop}%
\bibitem [{\citenamefont {Dziarmaga}(2010)}]{dziarmaga2010dynamics}%
  \BibitemOpen
  \bibfield  {author} {\bibinfo {author} {\bibfnamefont {J.}~\bibnamefont
  {Dziarmaga}},\ }\href@noop {} {\bibfield  {journal} {\bibinfo  {journal}
  {Advances in Physics}\ }\textbf {\bibinfo {volume} {59}},\ \bibinfo {pages}
  {1063} (\bibinfo {year} {2010})}\BibitemShut {NoStop}%
\bibitem [{\citenamefont {Polkovnikov}\ \emph {et~al.}(2011)\citenamefont
  {Polkovnikov}, \citenamefont {Sengupta}, \citenamefont {Silva},\ and\
  \citenamefont {Vengalattore}}]{polkovnikov2011colloquium}%
  \BibitemOpen
  \bibfield  {author} {\bibinfo {author} {\bibfnamefont {A.}~\bibnamefont
  {Polkovnikov}}, \bibinfo {author} {\bibfnamefont {K.}~\bibnamefont
  {Sengupta}}, \bibinfo {author} {\bibfnamefont {A.}~\bibnamefont {Silva}},\
  and\ \bibinfo {author} {\bibfnamefont {M.}~\bibnamefont {Vengalattore}},\
  }\href@noop {} {\bibfield  {journal} {\bibinfo  {journal} {Reviews of Modern
  Physics}\ }\textbf {\bibinfo {volume} {83}},\ \bibinfo {pages} {863}
  (\bibinfo {year} {2011})}\BibitemShut {NoStop}%
\bibitem [{\citenamefont {Greiner}\ \emph {et~al.}(2002)\citenamefont
  {Greiner}, \citenamefont {Mandel}, \citenamefont {H{\"a}nsch},\ and\
  \citenamefont {Bloch}}]{greiner2002collapse}%
  \BibitemOpen
  \bibfield  {author} {\bibinfo {author} {\bibfnamefont {M.}~\bibnamefont
  {Greiner}}, \bibinfo {author} {\bibfnamefont {O.}~\bibnamefont {Mandel}},
  \bibinfo {author} {\bibfnamefont {T.~W.}\ \bibnamefont {H{\"a}nsch}},\ and\
  \bibinfo {author} {\bibfnamefont {I.}~\bibnamefont {Bloch}},\ }\href@noop {}
  {\bibfield  {journal} {\bibinfo  {journal} {Nature}\ }\textbf {\bibinfo
  {volume} {419}},\ \bibinfo {pages} {51} (\bibinfo {year} {2002})}\BibitemShut
  {NoStop}%
\bibitem [{\citenamefont {Will}\ and\ \citenamefont
  {Will}(2013)}]{will2013coherent}%
  \BibitemOpen
  \bibfield  {author} {\bibinfo {author} {\bibfnamefont {S.}~\bibnamefont
  {Will}}\ and\ \bibinfo {author} {\bibfnamefont {S.}~\bibnamefont {Will}},\
  }\href@noop {} {\bibfield  {journal} {\bibinfo  {journal} {From Atom Optics
  to Quantum Simulation: Interacting Bosons and Fermions in Three-Dimensional
  Optical Lattice Potentials}\ ,\ \bibinfo {pages} {209}} (\bibinfo {year}
  {2013})}\BibitemShut {NoStop}%
\bibitem [{\citenamefont {Will}\ \emph {et~al.}(2015)\citenamefont {Will},
  \citenamefont {Iyer},\ and\ \citenamefont {Rigol}}]{will2015observation}%
  \BibitemOpen
  \bibfield  {author} {\bibinfo {author} {\bibfnamefont {S.}~\bibnamefont
  {Will}}, \bibinfo {author} {\bibfnamefont {D.}~\bibnamefont {Iyer}},\ and\
  \bibinfo {author} {\bibfnamefont {M.}~\bibnamefont {Rigol}},\ }\href@noop {}
  {\bibfield  {journal} {\bibinfo  {journal} {Nature communications}\ }\textbf
  {\bibinfo {volume} {6}},\ \bibinfo {pages} {6009} (\bibinfo {year}
  {2015})}\BibitemShut {NoStop}%
\bibitem [{\citenamefont {Gring}\ \emph {et~al.}(2012)\citenamefont {Gring},
  \citenamefont {Kuhnert}, \citenamefont {Langen}, \citenamefont {Kitagawa},
  \citenamefont {Rauer}, \citenamefont {Schreitl}, \citenamefont {Mazets},
  \citenamefont {Smith}, \citenamefont {Demler},\ and\ \citenamefont
  {Schmiedmayer}}]{gring2012relaxation}%
  \BibitemOpen
  \bibfield  {author} {\bibinfo {author} {\bibfnamefont {M.}~\bibnamefont
  {Gring}}, \bibinfo {author} {\bibfnamefont {M.}~\bibnamefont {Kuhnert}},
  \bibinfo {author} {\bibfnamefont {T.}~\bibnamefont {Langen}}, \bibinfo
  {author} {\bibfnamefont {T.}~\bibnamefont {Kitagawa}}, \bibinfo {author}
  {\bibfnamefont {B.}~\bibnamefont {Rauer}}, \bibinfo {author} {\bibfnamefont
  {M.}~\bibnamefont {Schreitl}}, \bibinfo {author} {\bibfnamefont
  {I.}~\bibnamefont {Mazets}}, \bibinfo {author} {\bibfnamefont {D.~A.}\
  \bibnamefont {Smith}}, \bibinfo {author} {\bibfnamefont {E.}~\bibnamefont
  {Demler}},\ and\ \bibinfo {author} {\bibfnamefont {J.}~\bibnamefont
  {Schmiedmayer}},\ }\href@noop {} {\bibfield  {journal} {\bibinfo  {journal}
  {Science}\ }\textbf {\bibinfo {volume} {337}},\ \bibinfo {pages} {1318}
  (\bibinfo {year} {2012})}\BibitemShut {NoStop}%
\bibitem [{\citenamefont {Trotzky}\ \emph {et~al.}(2012)\citenamefont
  {Trotzky}, \citenamefont {Chen}, \citenamefont {Flesch}, \citenamefont
  {McCulloch}, \citenamefont {Schollw{\"o}ck}, \citenamefont {Eisert},\ and\
  \citenamefont {Bloch}}]{trotzky2012probing}%
  \BibitemOpen
  \bibfield  {author} {\bibinfo {author} {\bibfnamefont {S.}~\bibnamefont
  {Trotzky}}, \bibinfo {author} {\bibfnamefont {Y.-A.}\ \bibnamefont {Chen}},
  \bibinfo {author} {\bibfnamefont {A.}~\bibnamefont {Flesch}}, \bibinfo
  {author} {\bibfnamefont {I.~P.}\ \bibnamefont {McCulloch}}, \bibinfo {author}
  {\bibfnamefont {U.}~\bibnamefont {Schollw{\"o}ck}}, \bibinfo {author}
  {\bibfnamefont {J.}~\bibnamefont {Eisert}},\ and\ \bibinfo {author}
  {\bibfnamefont {I.}~\bibnamefont {Bloch}},\ }\href@noop {} {\bibfield
  {journal} {\bibinfo  {journal} {Nature physics}\ }\textbf {\bibinfo {volume}
  {8}},\ \bibinfo {pages} {325} (\bibinfo {year} {2012})}\BibitemShut {NoStop}%
\bibitem [{\citenamefont {Broomhead}\ and\ \citenamefont
  {King}(1986)}]{broomhead1986extracting}%
  \BibitemOpen
  \bibfield  {author} {\bibinfo {author} {\bibfnamefont {D.~S.}\ \bibnamefont
  {Broomhead}}\ and\ \bibinfo {author} {\bibfnamefont {G.~P.}\ \bibnamefont
  {King}},\ }\href@noop {} {\bibfield  {journal} {\bibinfo  {journal} {Phys.
  D}\ }\textbf {\bibinfo {volume} {20}},\ \bibinfo {pages} {217} (\bibinfo
  {year} {1986})}\BibitemShut {NoStop}%
\bibitem [{\citenamefont {Packard}\ \emph {et~al.}(1980)\citenamefont
  {Packard}, \citenamefont {Crutchfield}, \citenamefont {Farmer},\ and\
  \citenamefont {Shaw}}]{packard1980geometry}%
  \BibitemOpen
  \bibfield  {author} {\bibinfo {author} {\bibfnamefont {N.~H.}\ \bibnamefont
  {Packard}}, \bibinfo {author} {\bibfnamefont {J.~P.}\ \bibnamefont
  {Crutchfield}}, \bibinfo {author} {\bibfnamefont {J.~D.}\ \bibnamefont
  {Farmer}},\ and\ \bibinfo {author} {\bibfnamefont {R.~S.}\ \bibnamefont
  {Shaw}},\ }\href@noop {} {\bibfield  {journal} {\bibinfo  {journal} {Phys.
  Rev. Lett.}\ }\textbf {\bibinfo {volume} {45}},\ \bibinfo {pages} {712}
  (\bibinfo {year} {1980})}\BibitemShut {NoStop}%
\bibitem [{\citenamefont {Pan}\ and\ \citenamefont
  {Duraisamy}(2020)}]{Pan2020}%
  \BibitemOpen
  \bibfield  {author} {\bibinfo {author} {\bibfnamefont {S.}~\bibnamefont
  {Pan}}\ and\ \bibinfo {author} {\bibfnamefont {K.}~\bibnamefont
  {Duraisamy}},\ }\href@noop {} {\bibfield  {journal} {\bibinfo  {journal}
  {Chaos}\ }\textbf {\bibinfo {volume} {30}},\ \bibinfo {pages} {073135}
  (\bibinfo {year} {2020})}\BibitemShut {NoStop}%
\bibitem [{\citenamefont {Takens}(1981)}]{Taken}%
  \BibitemOpen
  \bibfield  {author} {\bibinfo {author} {\bibfnamefont {F.}~\bibnamefont
  {Takens}},\ }in\ \href@noop {} {\emph {\bibinfo {booktitle} {Dynamical
  systems and turbulence, Warwick 1980}}}\ (\bibinfo  {publisher} {Springer},\
  \bibinfo {year} {1981})\ pp.\ \bibinfo {pages} {366--381}\BibitemShut
  {NoStop}%
\bibitem [{\citenamefont {Le~Clainche}\ and\ \citenamefont
  {Vega}(2017)}]{HODMD}%
  \BibitemOpen
  \bibfield  {author} {\bibinfo {author} {\bibfnamefont {S.}~\bibnamefont
  {Le~Clainche}}\ and\ \bibinfo {author} {\bibfnamefont {J.~M.}\ \bibnamefont
  {Vega}},\ }\href@noop {} {\bibfield  {journal} {\bibinfo  {journal} {SIAM J.
  Appl. Dyn. Syst.}\ }\textbf {\bibinfo {volume} {16}},\ \bibinfo {pages} {882}
  (\bibinfo {year} {2017})}\BibitemShut {NoStop}%
\bibitem [{\citenamefont {Kaneko}\ \emph {et~al.}(2025)\citenamefont {Kaneko},
  \citenamefont {Imada}, \citenamefont {Kabashima},\ and\ \citenamefont
  {Ohtsuki}}]{kaneko2025forecasting}%
  \BibitemOpen
  \bibfield  {author} {\bibinfo {author} {\bibfnamefont {R.}~\bibnamefont
  {Kaneko}}, \bibinfo {author} {\bibfnamefont {M.}~\bibnamefont {Imada}},
  \bibinfo {author} {\bibfnamefont {Y.}~\bibnamefont {Kabashima}},\ and\
  \bibinfo {author} {\bibfnamefont {T.}~\bibnamefont {Ohtsuki}},\ }\href@noop
  {} {\bibfield  {journal} {\bibinfo  {journal} {Physical Review Research}\
  }\textbf {\bibinfo {volume} {7}},\ \bibinfo {pages} {013085} (\bibinfo {year}
  {2025})}\BibitemShut {NoStop}%
\bibitem [{\citenamefont {Yin}\ \emph {et~al.}(2023)\citenamefont {Yin},
  \citenamefont {Chan}, \citenamefont {da~Jornada}, \citenamefont {Qiu},
  \citenamefont {Yang},\ and\ \citenamefont {Louie}}]{DMDdiag}%
  \BibitemOpen
  \bibfield  {author} {\bibinfo {author} {\bibfnamefont {J.}~\bibnamefont
  {Yin}}, \bibinfo {author} {\bibfnamefont {Y.-h.}\ \bibnamefont {Chan}},
  \bibinfo {author} {\bibfnamefont {F.}~\bibnamefont {da~Jornada}}, \bibinfo
  {author} {\bibfnamefont {D.}~\bibnamefont {Qiu}}, \bibinfo {author}
  {\bibfnamefont {C.}~\bibnamefont {Yang}},\ and\ \bibinfo {author}
  {\bibfnamefont {S.~G.}\ \bibnamefont {Louie}},\ }\href
  {https://www.sciencedirect.com/science/article/pii/S0021999123000049}
  {\bibfield  {journal} {\bibinfo  {journal} {J. Comput. Phys.}\ }\textbf
  {\bibinfo {volume} {477}},\ \bibinfo {pages} {111909} (\bibinfo {year}
  {2023})}\BibitemShut {NoStop}%
\bibitem [{\citenamefont {Yin}\ \emph {et~al.}(2022)\citenamefont {Yin},
  \citenamefont {Chan}, \citenamefont {da~Jornada}, \citenamefont {Qiu},
  \citenamefont {Louie},\ and\ \citenamefont {Yang}}]{DMDtwotime}%
  \BibitemOpen
  \bibfield  {author} {\bibinfo {author} {\bibfnamefont {J.}~\bibnamefont
  {Yin}}, \bibinfo {author} {\bibfnamefont {Y.-h.}\ \bibnamefont {Chan}},
  \bibinfo {author} {\bibfnamefont {F.}~\bibnamefont {da~Jornada}}, \bibinfo
  {author} {\bibfnamefont {D.}~\bibnamefont {Qiu}}, \bibinfo {author}
  {\bibfnamefont {S.~G.}\ \bibnamefont {Louie}},\ and\ \bibinfo {author}
  {\bibfnamefont {C.}~\bibnamefont {Yang}},\ }\href
  {https://doi.org/https://doi.org/10.1016/j.jocs.2022.101843} {\bibfield
  {journal} {\bibinfo  {journal} {J. Comput. Sci.}\ }\textbf {\bibinfo {volume}
  {64}},\ \bibinfo {pages} {101843} (\bibinfo {year} {2022})}\BibitemShut
  {NoStop}%
\bibitem [{\citenamefont {Li}\ \emph {et~al.}(2017)\citenamefont {Li},
  \citenamefont {Dietrich}, \citenamefont {Bollt},\ and\ \citenamefont
  {Kevrekidis}}]{li2017extended}%
  \BibitemOpen
  \bibfield  {author} {\bibinfo {author} {\bibfnamefont {Q.}~\bibnamefont
  {Li}}, \bibinfo {author} {\bibfnamefont {F.}~\bibnamefont {Dietrich}},
  \bibinfo {author} {\bibfnamefont {E.~M.}\ \bibnamefont {Bollt}},\ and\
  \bibinfo {author} {\bibfnamefont {I.~G.}\ \bibnamefont {Kevrekidis}},\
  }\href@noop {} {\bibfield  {journal} {\bibinfo  {journal} {Chaos: An
  Interdisciplinary Journal of Nonlinear Science}\ }\textbf {\bibinfo {volume}
  {27}} (\bibinfo {year} {2017})}\BibitemShut {NoStop}%
\bibitem [{\citenamefont {Williams}\ \emph {et~al.}(2015)\citenamefont
  {Williams}, \citenamefont {Rowley},\ and\ \citenamefont
  {Kevrekidis}}]{Matthew_edmd}%
  \BibitemOpen
  \bibfield  {author} {\bibinfo {author} {\bibfnamefont {M.~O.}\ \bibnamefont
  {Williams}}, \bibinfo {author} {\bibfnamefont {C.~W.}\ \bibnamefont
  {Rowley}},\ and\ \bibinfo {author} {\bibfnamefont {I.~G.}\ \bibnamefont
  {Kevrekidis}},\ }\href {https://doi.org/10.3934/jcd.2015005} {\bibfield
  {journal} {\bibinfo  {journal} {Journal of Computational Dynamics}\ }\textbf
  {\bibinfo {volume} {2}},\ \bibinfo {pages} {247} (\bibinfo {year}
  {2015})}\BibitemShut {NoStop}%
\bibitem [{\citenamefont {Lusch}\ \emph {et~al.}(2018)\citenamefont {Lusch},
  \citenamefont {Kutz},\ and\ \citenamefont {Brunton}}]{lusch2018deep}%
  \BibitemOpen
  \bibfield  {author} {\bibinfo {author} {\bibfnamefont {B.}~\bibnamefont
  {Lusch}}, \bibinfo {author} {\bibfnamefont {J.~N.}\ \bibnamefont {Kutz}},\
  and\ \bibinfo {author} {\bibfnamefont {S.~L.}\ \bibnamefont {Brunton}},\
  }\href@noop {} {\bibfield  {journal} {\bibinfo  {journal} {Nature
  communications}\ }\textbf {\bibinfo {volume} {9}},\ \bibinfo {pages} {4950}
  (\bibinfo {year} {2018})}\BibitemShut {NoStop}%
\bibitem [{\citenamefont {Haq}\ \emph {et~al.}(2022)\citenamefont {Haq},
  \citenamefont {Iwata},\ and\ \citenamefont {Kawahara}}]{haq2022dynamic}%
  \BibitemOpen
  \bibfield  {author} {\bibinfo {author} {\bibfnamefont {I.~U.}\ \bibnamefont
  {Haq}}, \bibinfo {author} {\bibfnamefont {T.}~\bibnamefont {Iwata}},\ and\
  \bibinfo {author} {\bibfnamefont {Y.}~\bibnamefont {Kawahara}},\ }\href@noop
  {} {\bibfield  {journal} {\bibinfo  {journal} {Computer Vision and Image
  Understanding}\ }\textbf {\bibinfo {volume} {216}},\ \bibinfo {pages}
  {103355} (\bibinfo {year} {2022})}\BibitemShut {NoStop}%
\bibitem [{\citenamefont {Shen}\ \emph {et~al.}(2025)\citenamefont {Shen},
  \citenamefont {Camps}, \citenamefont {Darbha}, \citenamefont {Szasz},
  \citenamefont {Klymko}, \citenamefont {Tubman}, \citenamefont {Van~Beeumen}
  \emph {et~al.}}]{shen2023estimating}%
  \BibitemOpen
  \bibfield  {author} {\bibinfo {author} {\bibfnamefont {Y.}~\bibnamefont
  {Shen}}, \bibinfo {author} {\bibfnamefont {D.}~\bibnamefont {Camps}},
  \bibinfo {author} {\bibfnamefont {S.}~\bibnamefont {Darbha}}, \bibinfo
  {author} {\bibfnamefont {A.}~\bibnamefont {Szasz}}, \bibinfo {author}
  {\bibfnamefont {K.}~\bibnamefont {Klymko}}, \bibinfo {author} {\bibfnamefont
  {N.~M.}\ \bibnamefont {Tubman}}, \bibinfo {author} {\bibfnamefont
  {R.}~\bibnamefont {Van~Beeumen}}, \emph {et~al.},\ }\href
  {https://doi.org/10.22331/q-2025-08-27-1836} {\bibfield  {journal} {\bibinfo
  {journal} {Quantum}\ }\textbf {\bibinfo {volume} {9}},\ \bibinfo {pages}
  {1836} (\bibinfo {year} {2025})}\BibitemShut {NoStop}%
\end{thebibliography}%


%


%


%


%

\end{document}